\documentclass[sigconf, nonacm]{acmart}
\usepackage{pvldb}

\usepackage{amsmath}
\usepackage{amsthm}
\usepackage{booktabs}
\usepackage{graphicx}
\usepackage{placeins}
\usepackage{enumitem}
\usepackage{xspace}
\usepackage[ruled,vlined,linesnumbered]{algorithm2e}

\newtheorem{theorem}{Theorem}[section]
\newtheorem{lemma}[theorem]{Lemma}
\newtheorem{corollary}[theorem]{Corollary}

\newcommand{\gpk}{\textsc{GPKmeans}\xspace}
\newcommand{\ours}{\textsc{GPKmeans}\xspace}
\newcommand{\LSCK}{\textsc{LSCK}\xspace}
\newcommand{\LSCKHC}{\textsc{LSCK-HC}\xspace}
\newcommand{\Lloyd}{\textsc{Lloyd}\xspace}
\newcommand{\SummaryCentroids}{\textsc{SummaryCentroids}\xspace}
\newcommand{\SummaryC}{\textsc{SummaryC}\xspace}
\newcommand{\PPLL}{\textsc{PPLL}\xspace}
\newcommand{\PCCC}{\textsc{PCCC}\xspace}
\newcommand{\COPKMeans}{\textsc{COPKMeans}\xspace}
\newcommand{\OracleClean}{\textsc{OracleClean}\xspace}

\newcommand{\OursLSCKConstraints}{\textsc{GPKmeans}-LSCKConstraints\xspace}

\newcommand{\Banking}{\textsc{Banking77}\xspace}
\newcommand{\CLINC}{\textsc{CLINC150}\xspace}
\newcommand{\Massive}{\textsc{MASSIVE}\xspace}
\newcommand{\MTOP}{\textsc{MTOP}\xspace}
\newcommand{\GoEmotions}{\textsc{GoEmotions}\xspace}
\newcommand{\CLINCI}{\textsc{CLINC}(I)\xspace}
\newcommand{\MassiveI}{\textsc{Massive}(I)\xspace}
\newcommand{\MTOPI}{\textsc{MTOP}(I)\xspace}
\newcommand{\CLINCD}{\textsc{CLINC}(D)\xspace}
\newcommand{\MassiveD}{\textsc{Massive}(D)\xspace}
\newcommand{\MTOPD}{\textsc{MTOP}(D)\xspace}
\newcommand{\GoEmo}{\textsc{GoEmo}\xspace}
\newcommand{\Reddit}{\textsc{Reddit}\xspace}

\newcommand{\bX}{\boldsymbol{X}}
\newcommand{\bx}{\boldsymbol{x}}
\newcommand{\bF}{\boldsymbol{F}}
\newcommand{\bff}{\boldsymbol{f}}

\newcommand{\bz}{\boldsymbol{z}}
\newcommand{\bmu}{\boldsymbol{\mu}}

\setcopyright{none}
\date{}

\begin{document}
\title{Efficient Group-Penalty \texorpdfstring{$k$}{k}-means for Text Clustering with Noisy LLM-Generated Constraints}

\author{%
  Feiteng Nie\textsuperscript{1},
  Yuan Sun\textsuperscript{2},
  Shengkun Zhu\textsuperscript{1},
  Sheng Wang\textsuperscript{1*},
  Zhiyong Peng\textsuperscript{1}}
\affiliation[obeypunctuation=true]{%
  \institution{\textsuperscript{1}School of Computer Science, Wuhan University}, \country{China};
  \institution{\textsuperscript{2}La Trobe Business School, La Trobe University}, \country{Australia}}
\email{\{niefeiteng,whuzsk66,swangcs,peng\}@whu.edu.cn, yuan.sun@latrobe.edu.au}

\renewcommand{\authors}{Feiteng Nie, Yuan Sun, Shengkun Zhu, Sheng Wang, and Zhiyong Peng}
\renewcommand{\shortauthors}{Nie et al.}

\begin{abstract}
$k$-means provides a simple and efficient foundation for text clustering, but embedding alone can miss task-specific semantic boundaries.
Must-link (ML) and cannot-link (CL) constraints can correct such errors, yet conventional constrained clustering assumes that these relations are already supplied.
LLMs can generate them without task-specific labels, but exhaustive querying is costly and their judgments can be incorrect.
We present \gpk, an efficient Group-Penalty $k$-means method for text clustering with noisy LLM-generated constraints.
\gpk uses LLM as an offline local pairwise annotator, querying selected anchor-candidate pairs for sparse ML links and constructing set-valued CL evidence without all-pairs queries.
It converts connected components of the ML-link graph and the generated CL sets into groups with normalized, bounded soft penalties, limiting each group's influence independently of its size.
Cached coordinate descent is proposed to efficiently solve the constraints clustering problem.
We prove finite-step convergence and derive deterministic and probabilistic oracle-referenced bounds that relate objective distortion to the penalty coefficients and to the group-level error probabilities of the generated constraints.
Across nine datasets, \gpk uses fewer LLM calls and yields more accurate constraints.
It achieves the best average ACC, NMI, and ARI with an approximate $30\times$ speedup over state-of-the-art constrained clustering methods.

\end{abstract}

\maketitle
\pagestyle{plain}
\thispagestyle{plain}

\begingroup
\renewcommand{\thefootnote}{\fnsymbol{footnote}}
\footnotetext[1]{Sheng Wang is the corresponding author.}
\endgroup

\section{Introduction}

Text clustering is an important primitive for organizing and analyzing textual records in data management, with applications including topic discovery, recommendation, anomaly detection, and community analysis~\cite{Meng2022,Chandola2009,Su2024}.
In a typical pipeline, pretrained encoders first map texts into a continuous embedding space, after which $k$-means partitions the resulting representations in a simple and scalable manner~\cite{Liu2019,su2023instructor,lloyd1982least,bahmani2012scalable,Wang20On}. 
However, geometric proximity in the embedding space does not always coincide with task-specific semantic similarity. 
Texts expressing distinct intents may exhibit nearly identical lexical patterns and highly similar embeddings, whereas texts expressing the same intent may differ substantially in surface form, and $k$-means lacks the external evidence needed to identify and correct such errors.

Pairwise side information provides a natural source of such evidence. 
A must-link (ML) indicates that two samples should be assigned to the same cluster, whereas a cannot-link (CL)  indicates that they should be separated. 
Constrained-clustering methods incorporate these relations to guide the clustering process~\cite{wagstaff2001constrained,basu2004probabilistic,bilenko2004integrating,basu2008constrained,cheng2008constrained}. 
Traditionally, pairwise relations are elicited from human annotators and domain experts~\cite{basu2004active,basu2008constrained,WangKFF12}, while benchmark evaluations often construct them from ground-truth labels~\cite{wagstaff2001constrained,basu2004probabilistic,bilenko2004integrating}. 
However, existing methods assume predefined constraints and do not address how to acquire them without labels or human supervision.

LLMs have shown promise in data cleaning and integration~\cite{NarayanCOR22}, motivating their use to acquire semantic relations from text without task-specific gold labels~\cite{achiam2023gpt4,brown2020fewshot,dubey2024llama}.
Consider a small example containing two refund-related intents in Figure~\ref{fig:example1}.
Although the texts share similar refund vocabulary, the intended user actions differ, so embedding-based $k$-means can cross the local semantic boundary and assign a delayed-refund query to the refund-request cluster.
Local LLM constraints can help reveal this boundary by linking semantically matched examples.
However, the same lexical overlap can also induce an incorrect cross-intent ML constraint, pulling the assignment in the wrong direction.
Thus, even $k$-means with generated ML and CL constraints can produce undesirable results when those constraints are treated as hard rules.

\begin{figure}[t]
    \centering
    \includegraphics[width=\linewidth]{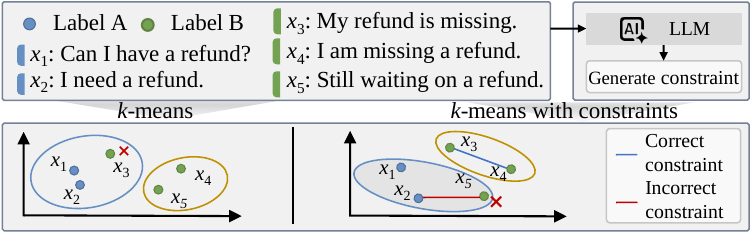}
    \caption{An example from \Banking with two refund-related intents.
    $k$-means can misplace a text across the two intents, and LLM-generated pairwise constraints can include both useful and incorrect must-link constraints.}
    \label{fig:example1}
    \vspace{-2.6em}
\end{figure}

LLMs provide semantic guidance for text clustering through label generation, representation refinement, and textual cluster descriptions or centroids~\cite{zhang2023clusterllm,viswanathan2024fewshot,zhang2025ppll,rodriguez2025kllmmeans,feng2024llmedgerefine}.
However, directly assigning clusters with LLMs can produce inconsistent labels and semantic boundaries across batches and iterations~\cite{mu2024addressing}.
When the corpus exceeds the context window, no single call can observe all samples, making it difficult to enforce cross-batch consistency and making repeated inference costly~\cite{fu2025context}. 
Using an LLM as a local pairwise annotator offers a more focused alternative, but it shifts rather than eliminates the difficulty. 
Exhaustive pair querying remains prohibitively expensive, so a limited query budget can be allocated to a part of the selected data; the resulting relations are globally sparse yet may be locally redundant or erroneous~\cite{AddankiGS2021How}. 
The resulting data management challenge is to improve semantic grouping under a limited query budget while controlling the influence of noisy relations and the computational cost of clustering.
\LSCK uses LLM-generated ML and CL constraints to correct local errors in embedding-based clustering~\cite{jia2026optimized}, while other methods use cluster descriptions, few-shot guidance, label inference, or summary centroids to assist clustering~\cite{zhang2023clusterllm,viswanathan2024fewshot,zhang2025ppll,rodriguez2025kllmmeans,feng2024llmedgerefine}.
However, strong enforcement of noisy LLM-generated constraints can distort clustering, while computationally expensive constrained optimization can limit scalability.

Various algorithms have been proposed for clustering with hard or soft ML and CL constraints~\cite{wagstaff2001constrained,basu2004probabilistic,bilenko2004integrating,davidson2005clustering,jia2023efficient,baumann2025pccc}.
\COPKMeans can fail to produce a valid partition when no feasible assignment is found~\cite{wagstaff2001constrained,davidson2005clustering}.
\PCCC accommodates uncertain relations through soft penalties, but solves mixed-integer assignment problems that can incur substantial computational cost~\cite{baumann2025pccc}.
These limitations motivate a method that tolerates noisy constraints while retaining efficient local updates.

To address this challenge, we propose \gpk, an efficient Group-Penalty $k$-means method for organizing text collections using sparse, noisy LLM-generated constraints.
\gpk follows three stages. 
First, it selects anchors within preliminary $k$-means clusters by favoring large centroid margins and proximity to their assigned centroids, queries their neighbors for sparse ML constraints, and constructs set-valued CL constraints without querying all text pairs.
Second, it converts connected components of the generated ML-link graph and the generated CL sets into typed constraint groups and incorporates them into the $k$-means objective through normalized soft penalties. 
Third, \gpk efficiently solves the resulting constraints $k$-means problem using cached coordinate descent. 
Cached cluster and group statistics allow candidate reassignments to be evaluated from only the affected objective terms, avoiding repeated global constrained optimization.
Together with bounded penalties, \gpk exploits reliable constraints while limiting the influence of noisy groups.

Our contributions are as follows:
\begin{itemize}
\item We propose \gpk, an efficient Group-Penalty $k$-means method that incorporates noisy LLM-generated ML and CL constraints into the assignment objective through normalized, bounded soft penalties rather than hard rules.

\item We develop an anchor-based querying strategy that constructs connected-component ML groups and uses cached coordinate descent to efficiently solve the constraints clustering problem.

\item We establish finite-step convergence and derive bounds that quantify how noisy constraints affect the objective as the ML and CL penalty coefficients vary. 

\item We demonstrate the effectiveness and efficiency of \gpk over nine datasets in our experimental evaluations. 
Our experimental findings indicate that \gpk achieves the highest averaged ACC, NMI, and ARI, while running approximately $30\times$ faster than the fastest constrained-clustering baseline considered.
\end{itemize}

\section{Related Work}

\subsection{Text clustering.}
$k$-means is a standard clustering primitive because it is simple, scalable, and tied to a clear squared-error objective~\cite{macqueen1967kmeans,lloyd1982least,jain2010data,xu2005survey,bahmani2012scalable,zhang1996birch,guha1998cure,ng1994efficient,Wang20On}.
Modern text clustering commonly applies these objectives after mapping documents into embedding spaces such as Sentence-BERT~\cite {reimers2019sentencebert}, INSTRUCTOR-large~\cite{su2023instructor}, instruction-tuned embeddings, or neural topic representations~\cite{xie2016unsupervised,angelov2020top2vec,grootendorst2022bertopic,aggarwal2012survey,muennighoff2023mteb}.
However, embedding distance is only a proxy for the desired task partition: nearby texts may share entities or domains while requiring different intent labels, and distant texts may express the same intent with different wording.
\gpk keeps this efficient embedding-plus-$k$-means backbone, but adds a local semantic correction channel rather than learning a new representation.

\subsection{Constrained clustering.}
Constrained clustering augments an unsupervised objective with must-link and cannot-link side information~\cite{cheng2008constrained}.
\COPKMeans treats these constraints as hard feasibility requirements~\cite{wagstaff2001constrained,klein2002instance}, while probabilistic formulations, metric-learning methods, and confidence-based constrained clustering relax supervision into soft terms, learned distances, spectral objectives, or weighted ML and CL constraints~\cite{basu2004probabilistic,bilenko2004integrating,xing2003distance,kulis2013metric,barhillel2003learning,shental2004computing,kamvar2003spectral,basu2004active,basu2008constrained,chapelle2006semi}.
\PCCC employs integer programming to solve constrained clustering problems~\cite{baumann2025pccc}.
Hard constraints can become infeasible or brittle: one incorrect ML constraint can permanently bind points that should remain separable, and one incorrect CL constraint can force an unnecessary split~\cite{davidson2005clustering}.
Flexible constrained spectral clustering and related soft formulations treat constraints as preferences rather than absolute rules~\cite{wang2010flexible}, which is closer to our setting.
However, existing methods require computationally intensive machinery, such as integer programming, branch-and-bound, or semidefinite optimization, resulting in high computational complexity.

\subsection{LLM-based clustering.}
LLMs have been used as flexible semantic annotators~\cite{brown2020fewshot,gilardi2023chatgpt,NarayanCOR22}, connecting LLM-guided clustering to weak supervision and learning with noisy labels~\cite{ratner2020snorkel,bach2019snorkel,natarajan2013learning}.
Recent LLM-based clustering work asks LLMs to guide text clustering through pairwise preferences, few-shot cluster descriptions, boundary refinement, LLM-derived positive pairs, or summary centroids~\cite{zhang2023clusterllm,viswanathan2024fewshot,feng2024llmedgerefine,zhang2025ppll,rodriguez2025kllmmeans}.
The closest line to our setting is LLM-generated constrained clustering, including \LSCK, which uses LLM-produced ML and CL constraints as hard or high-strength guidance~\cite{jia2026optimized}.
However, it depends on combinatorial optimization techniques, which limits its scalability to large problem instances.

\noindent{\textbf{Remarks.}}
\gpk uses the LLM in a narrower and more reproducible role: LLM calls are made offline to generate ML and CL constraints, and replayed as inputs to a $k$-means.
The key distinction is that \gpk treats these constraints as bounded ML-group and pairwise CL regularizers and optimizes them with cached exact local-delta updates, rather than using hard rules or direct LLM-generated labels.

\section{Preliminaries}
\label{sec:prelim}

This section introduces the assignment-matrix formulation of $k$-means used throughout the paper and its equivalent centroid-free objective.
We then define pairwise ML and CL constraints and clarify how LLM-generated constraints enter the clustering problem as fixed inputs.
Table~\ref{tab:symbols} summarizes the main notation.

\subsection{Problem Formulation of $k$-means}
\label{sec:prelim:setup}

Consider a text collection $\mathcal{D}=\{x_i\}_{i=1}^n$ and precomputed embeddings $\bx_i\in\mathbb{R}^d$, where $[n]=\{1,\ldots,n\}$ and $[k]=\{1,\ldots,k\}$.
Let $\bX=[\bx_1,\ldots,\bx_n]\in\mathbb{R}^{d\times n}$ denote the embedding matrix, and assume that the number of clusters $k$ is given.
We represent a clustering by a one-hot indicator matrix $\bF=[\bff_1,\ldots,\bff_k]\in\{0,1\}^{n\times k}$, where each row has exactly one nonzero entry and each column is nonempty.
Let $\mathcal{F}_{n,k}$ denote the set of such nonempty one-hot assignment matrices.
The induced label is $z_i=\arg\max_{\ell\in[k]} \bF_{i\ell}$, and we use $\bz$ only as shorthand for the partition encoded by $\bF$.

Let $\mathbf{M}=[\bmu_1,\ldots,\bmu_k]\in\mathbb{R}^{d\times k}$ collect the centroids.
The standard $k$-means SSE is
\begin{equation}
  \min_{\bF\in\mathcal{F}_{n,k},\mathbf{M}}\;
  \sum_{\ell=1}^k \sum_{i:\bF_{i\ell}=1} \|\bx_i-\bmu_\ell\|_2^2 .
  \label{eq:kmeans-sse}
\end{equation}
For a fixed assignment $\bF$, the centroid variables have a closed-form minimizer.
Since $\bff_\ell$ selects the samples assigned to cluster~$\ell$, $\bff_\ell^\top\bff_\ell$ is the cluster size and $\bX\bff_\ell$ is the corresponding vector sum.
Therefore the optimal centroid is the sample mean
\begin{equation}
  \bmu_\ell^* = \frac{\bX\bff_\ell}{\bff_\ell^\top\bff_\ell},\qquad \ell\in[k].
  \label{eq:optimal-centroid}
\end{equation}
Since each row of $\bF$ contains only one nonzero entry, $\bF^\top\bF$ is diagonal, and its $\ell$-th diagonal entry is $\bff_\ell^\top\bff_\ell$.
Substituting the optimal centroids back into Eq.~\eqref{eq:kmeans-sse} removes the centroid variables and leaves an objective depending only on $\bF$~\cite{Nie22Coordinate}.
Therefore, Eq.~\eqref{eq:kmeans-sse} can be equivalently rewritten as
\begin{equation}
\begin{aligned}
  S(\bF)
  &= C-\sum_{\ell=1}^k
  \frac{\bff_\ell^\top\bX^\top\bX\bff_\ell}{\bff_\ell^\top\bff_\ell},
  \qquad C=\sum_{i=1}^n\|\bx_i\|_2^2 .
\end{aligned}
  \label{eq:sse}
\end{equation}
The constant $C$ does not depend on the clustering, so minimizing $S(\bF)$ is equivalent to maximizing the centroid-free score
\begin{equation}
  \mathrm{obj}(\bF)
  := \sum_{\ell=1}^k
  \frac{\bff_\ell^\top\bX^\top\bX\bff_\ell}{\bff_\ell^\top\bff_\ell} .
  \label{eq:sse-equivalence}
\end{equation}
Therefore, solving Eq.~\eqref{eq:kmeans-sse} is equivalent to maximizing Eq.~\eqref{eq:sse-equivalence} over $\mathcal{F}_{n,k}$, whose only optimization variable is the assignment matrix $\bF$.

\begin{table}[t]
\centering
\caption{Summary of notations.}
\vspace{-1.4em}
\label{tab:symbols}
\renewcommand{\arraystretch}{1.08}
\begin{tabular}{@{}p{0.25\linewidth}p{0.67\linewidth}@{}}
\toprule
Notation & Description \\
\midrule
$\bX\in\mathbb{R}^{d\times n}$ & The embedding matrix \\
$\bF\in\{0,1\}^{n\times k}$ & The indicator matrix \\
$\mathcal{F}_{n,k}$ & The nonempty one-hot assignment space \\
$\bz\in[k]^n$ & The induced cluster labels \\
$S(\cdot)$ & The $k$-means SSE \\
$\mathcal{M}$ & The collection of ML constraint groups \\
$\mathcal{C}$ & The collection of CL constraint groups \\
$G_t$ & A constraint group \\
$\phi_{\mathrm{ML}},\phi_{\mathrm{CL}}$ & Normalized ML/CL-group penalties \\
$\lambda_{\mathrm{ML}},\lambda_{\mathrm{CL}}\ge0$ & Global ML and CL penalty coefficients \\
\bottomrule
\vspace{-3.6em}
\end{tabular}
\end{table}

\subsection{ML and CL Constraints}
\label{sec:prelim:constraints}

In constrained clustering, side information is commonly expressed through must-link and cannot-link relations~\cite{wagstaff2001constrained,klein2002instance,basu2004probabilistic,bilenko2004integrating,basu2008constrained}.
A must-link (ML) constraint on a pair $\{i,j\}$ indicates that the two samples should be co-clustered, while a cannot-link (CL) constraint indicates that they should be assigned to different clusters.
The raw side information may be provided as individual pairs or as sets that jointly encode multiple pairwise relations.
In either form, the constrained relations are sparse and need not cover all $\binom{n}{2}$ possible pairs.

In traditional constrained clustering, ML and CL constraints are usually treated as externally provided information~\cite{wagstaff2001constrained,basu2004probabilistic,bilenko2004integrating}.
Such constraints may come from human annotations, expert rules, or existing domain knowledge, and are therefore often assumed to be relatively reliable.
In text clustering, however, manually annotating ML/CL constraints is expensive and difficult to scale.
This makes LLMs an attractive option for generating ML and CL  information from selected texts~\cite{jia2026optimized}.
The LLM is only responsible for constraint generation; the constrained clustering problem below treats the resulting ML and CL constraints as fixed inputs.
The next section describes how \gpk organizes these inputs into group-level soft regularizers and incorporate them into the $k$-means assignment objective.

\section{Algorithm Design}
\label{sec:method}

\begin{figure*}[t]
  \centering
  \includegraphics[width=\textwidth]{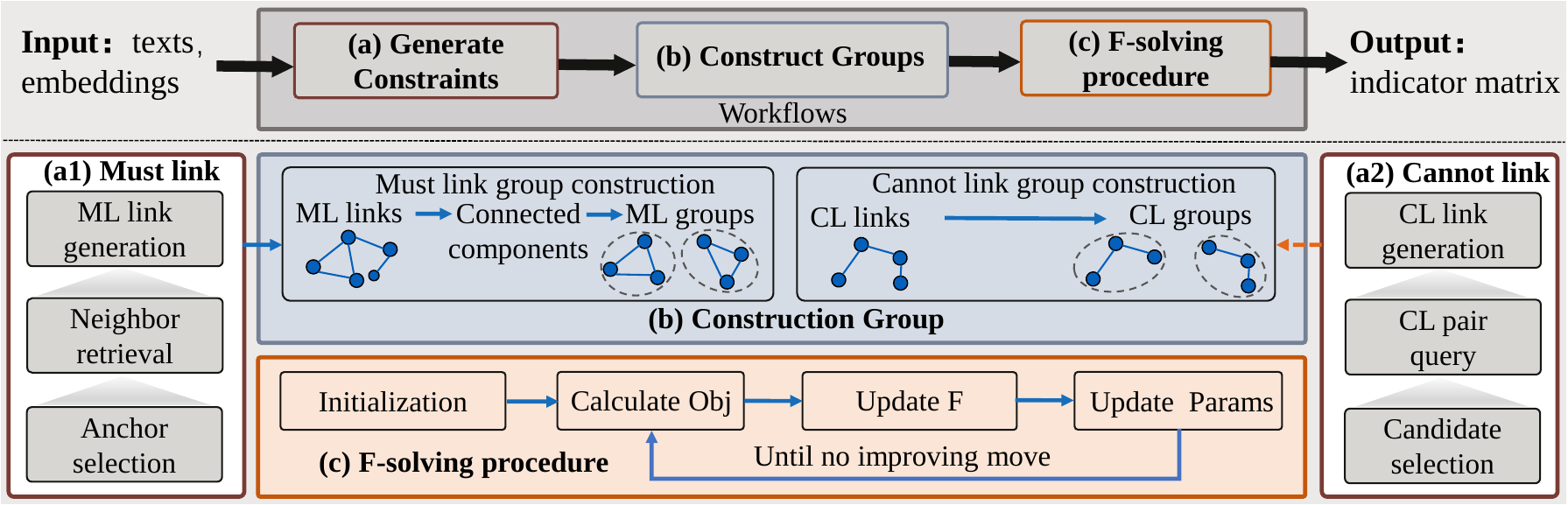}
  \vspace{-2.6em}
  \caption{Overview of \gpk.}
  \label{fig:pipeline}
  \vspace{-2em}
\end{figure*}

\subsection{Overview}
\label{sec:method:overview}

Given text samples $\mathcal{D}$, precomputed embeddings $\{\bx_i\}_{i=1}^n$, and a target cluster count $k$, \gpk efficiently returns a nonempty $k$-way assignment matrix $\bF\in\mathcal{F}_{n,k}$, equivalently labels $\bz\in[k]^n$.
The method is designed for the setting where LLM-generated constraints are useful but sparse, local, and potentially noisy.
Accordingly, ML and CL constraints enter as soft terms in a regularized $k$-means objective rather than as hard feasibility rules.

As shown in Figure~\ref{fig:pipeline}, \gpk proceeds in three stages.
First, offline constraint construction obtains sparse ML links and set-valued CL constraints without exhaustively querying all text pairs.
Second, the ML links and CL sets are represented as constraint groups with type-specific normalized penalties, so each group's contribution is bounded by the applicable global penalty coefficient.
Third, \gpk solves the assignment matrix efficiently with cached coordinate descent, maintaining cluster statistics and group occupancy counts to evaluate each candidate move by an exact local delta without recomputing the full objective or expanding ML groups into implied pairs.
For the main experiments, LLM calls are offline and the generated constraints are replayed.

\subsection{Constraint Generation}
\label{sec:method:anchor}

The anchor-based ML querying procedure selects $N_A$ real data points as anchors and retrieves up to $N_C$ nearby candidates per anchor.
We first run preliminary $k$-means on $\ell_2$-normalized embeddings and rank points within each cluster by $s_i=m_i-0.25d_i$, where $m_i$ is the gap between the point's squared distances to its two nearest centroids and $d_i$ is the squared distance to the assigned centroid.
We retain up to $q$ highest-scoring points per cluster, yielding $N_A\le qk$ anchors and favoring well-separated representatives near cluster centers.
This distributes anchors across preliminary clusters without using gold labels or LLM calls.
For each selected anchor, we query a small set of nearby candidates, excluding the anchor itself, and ask only about anchor-candidate pairs.
This creates a star-shaped local query pattern centered at the anchor, avoiding all-pairs queries among the candidates and reducing the number of LLM pair queries from $O(n^2)$ to $O(N_A N_C)$.
Using real text samples as anchors ensures the LLM receives grounded, self-contained prompts rather than abstract vector representations.

For each queried pair $(a,c)$, the LLM identifies whether the candidate clearly matches the anchor; accepted pairs become ML links, while uncertain or nonmatching candidates are not converted into CL constraints.
The anchor-based query pattern deliberately targets local ambiguity, because querying nearby candidates concentrates LLM effort on pairwise constraints that can change an assignment.
The generated constraints are stored as files; the main experiments replay these files without new LLM calls.
Algorithm~\ref{alg:anchor-querying} gives the full procedure.

\begin{algorithm}[t]
  \caption{Anchor-Based ML-Link Generation}
  \label{alg:anchor-querying}
  \KwIn{Text samples $\mathcal{D}$, embeddings $\{\mathbf{x}_i\}_{i=1}^n$,
        target clusters $k$, per-cluster anchor quota $q$, candidate count $N_C$, LLM model.}
  \KwOut{ML groups $\mathcal{M}$.}
  Run preliminary $k$-means on $\ell_2$-normalized embeddings\;
  $E_{\mathrm{ML}}\leftarrow\emptyset$\;
  \textcolor{gray}{\underline{\textit{Anchor selection:}}} Rank points within each cluster by $s_i=m_i-0.25d_i$ (Section~\ref{sec:method:anchor})\;
  Select up to $q$ highest-scoring points per cluster to form $\mathcal{A}$; set $N_A\leftarrow|\mathcal{A}|$\;
  \For{each anchor $a\in\mathcal{A}$}{
    \textcolor{gray}{\underline{\textit{Candidate retrieval:}}}
      $\mathrm{cands}(a)\leftarrow\mathrm{kNN}(\mathbf{x}_a,\{\mathbf{x}_i:i\ne a\},N_C)$\;
    \For{each candidate $c\in\mathrm{cands}(a)$}{
      \textcolor{gray}{\underline{\textit{LLM query:}}} Query whether $(a,c)$ is a clear semantic match;
      add only positive responses to $E_{\mathrm{ML}}$\;
    }
  }
  \textcolor{gray}{\underline{\textit{ML groups:}}} Build $H=(V,E_{\mathrm{ML}})$ and compute DSU components\;
    $\mathcal{M}\leftarrow\{G\in\mathrm{ConnectedComponents}(H):|G|\ge 2\}$\;
  \KwRet{$\mathcal{M}$}\;
  \end{algorithm}

For CL constraints, we adopt the set-valued construction of Jia et al.~\cite{jia2026optimized}.
Starting from an uncovered sample, the procedure grows a separation set $Y$ by uniformly sampling uncovered candidates whose distance from $Y$ exceeds the threshold $\rho=\mathrm{cost}_{kc}$, where $\mathrm{cost}_{kc}$ is the cost returned by the min-max $k$-center procedure.
The LLM adds a candidate only when it finds no semantic match among the current members; construction stops when $|Y|=k$ or no eligible candidate remains, and then restarts to form another set.
This produces a family of CL sets whose members are intended to occupy different clusters.
We reuse these generated sets as fixed side information, but not the matching-based CL solver of Jia et al.; instead, each set is converted into the normalized group penalty defined below and optimized jointly with the SSE and ML terms.

\subsection{Group Construction and Normalized Penalties}
\label{sec:method:groups}
\label{sec:method:obj}

The two constraint sources are mapped into a common group representation: connected components of the generated ML-link graph form the must-link collection $\mathcal{M}=\{G_t\}$, while each generated CL set directly forms one group in the cannot-link collection $\mathcal{C}=\{G_t\}$, where $G_t\subseteq[n]$.
Members of an ML group are encouraged to be co-clustered, while members of a CL group are discouraged from colliding in the same cluster.
This group view defines the objects used by the soft penalties below, without permanently merging samples or imposing graph-coloring constraints.

For component construction, each generated pairwise ML constraint is treated as an undirected link in a graph $H=(V,E_{\mathrm{ML}})$ over samples.
\gpk converts this graph into connected components via a disjoint-set union (DSU) procedure, and each component with at least two samples becomes an ML group $G\in\mathcal{M}$.
A component with $m$ samples is represented as a single group rather than expanded into $\binom{m}{2}$ implied pairwise penalties.
For a group $G$ with $m=|G|$ and cluster occupancy counts $n_\ell(G)=|\{i\in G:z_i=\ell\}|$, we define
\begin{equation}
\phi_{\mathrm{ML}}(G;\mathbf{z})
= \frac{m - \max_\ell\, n_\ell(G)}{m-1}, \quad m \ge 2.
\label{eq:ml}
\end{equation}
The numerator counts members outside the majority cluster; dividing by $m-1$ normalizes the penalty to $[0,1]$.
A perfectly co-clustered group incurs zero cost; the penalty is bounded by one and reaches one when every member occupies a distinct cluster, which is possible only when $m\le k$.
Thus, componentization is not hard merging: samples remain individually movable, and an incorrectly included sample contributes a bounded penalty instead of forcing the entire component to be merged.

Each separation set becomes a CL group $G\in\mathcal{C}$.
Unlike independent pairwise penalties, one normalized group term represents the mutually different relations induced by the set.
For a general CL group, we penalize the fraction of within-cluster pairs:
\begin{equation}
\phi_{\mathrm{CL}}(G;\mathbf{z})
= \frac{\displaystyle\sum_\ell \binom{n_\ell(G)}{2}}{\binom{m}{2}}, \quad m=|G|\ge 2.
\label{eq:cl}
\end{equation}
This is the fraction of pairs in $G$ that collide in the same cluster, normalized to $[0,1]$.
When $m=2$, Eq.~\eqref{eq:cl} reduces to the standard pairwise cannot-link penalty $\mathbb{I}[z_a=z_c]$.
Together with Eq.~\eqref{eq:ml}, this shows that the normalized group penalties recover classical pairwise ML and CL penalties as a special case.

If a group were expanded into all implied pairwise penalties without normalization, a group of size $m$ would produce $O(m^2)$ terms and allow large groups to dominate the objective.
With normalization, $\phi_{\mathrm{ML}},\phi_{\mathrm{CL}}\in[0,1]$, so the maximum objective shift from any incorrect group is bounded by the corresponding penalty coefficient regardless of its size (Lemma~\ref{lem:bounded}).
The coefficients $\lambda_{\mathrm{ML}}$ and $\lambda_{\mathrm{CL}}$ control the strength of the soft constraints.

The optimization variable is the one-hot assignment matrix $\bF$ introduced in Section~\ref{sec:prelim}; for readability, the optimizer is written in terms of the induced label vector $\bz$.
We write $S(\bz)$ for the same centroid-free SSE $S(\bF)$ in Eq.~\eqref{eq:sse} evaluated at the assignment matrix induced by $\bz$.
The complete \gpk objective is
\begin{equation}
  \begin{aligned}
  J(\bz)
  &\!=\!S(\bz)
  + \lambda_{\mathrm{ML}}\!\!\sum_{G\in\mathcal{M}}\!\! \phi_{\mathrm{ML}}(G;\bz) 
  + \lambda_{\mathrm{CL}}\!\!\sum_{G\in\mathcal{C}}\!\! \phi_{\mathrm{CL}}(G;\bz),
  \end{aligned}
  \label{eq:objective}
\end{equation}
The global coefficients $\lambda_{\mathrm{ML}}$ and $\lambda_{\mathrm{CL}}$ control how strongly the optimizer trades cluster compactness for semantic consistency.
In experiments, we use fixed default coefficients and report separate sensitivity studies for these choices.

\subsection{Solving the Assignment Matrix}
\label{sec:method:opt}

Given the constructed groups and the objective in Eq.~\eqref{eq:objective}, the final stage solves for the assignment matrix $\bF$.
Equivalently, we optimize its induced label vector $\bz$, where each coordinate update moves one sample from its current cluster to a candidate cluster.
\gpk uses a cached coordinate-descent solver for this discrete assignment problem.
Because the constraints are soft penalties, the solver only needs to test whether a candidate move decreases the regularized objective; it does not need to repair constraint violations or solve a separate global constrained subproblem.
Algorithm~\ref{alg:gpkmeans} summarizes the procedure.

\begin{algorithm}[t]
  \caption{\gpk Assignment Solver}
  \label{alg:gpkmeans}
  \KwIn{Embeddings $\{\mathbf{x}_i\}_{i=1}^n$, target clusters $k$, groups
    $\mathcal{M},\mathcal{C}$, coefficients $\lambda_{\mathrm{ML}},\lambda_{\mathrm{CL}}$.}
  \KwOut{Assignment $\mathbf{z}\in[k]^n$.}
  \textcolor{gray}{\underline{\textit{Initialization:}}}
  Initialize $\mathbf{z}$ using a $k$-means++ seeded $k$-means run\;
  Compute cluster statistics $a_\ell,\mathbf{b}_\ell$ for all $\ell\in[k]$ by Eq.~\eqref{eq:stats}\;
  Build the group-count cache $n_\ell(G)$ for all $\ell\in[k]$ and $G\in\mathcal{M}\cup\mathcal{C}$\;
  \BlankLine
  \textcolor{gray}{\underline{\textit{Coordinate descent:}}}
  \While{not converged}{
    \For{$i=1,\dots,n$}{
      $r\leftarrow z_i$\;
      \If{$a_r=1$}{continue\;}
      $\Delta J^* \leftarrow 0$;\quad $q^*\leftarrow r$\;
      \textcolor{gray}{\underline{\textit{Local move evaluation:}}}
      \For{$q\in[k]\setminus\{r\}$}{
        Compute $\Delta S_{i:r\to q}$ by Eq.~\eqref{eq:move-sse}\;
        Compute $\Delta J_{i:r\to q}$ by adding the incident-group ML/CL penalty deltas in Eq.~\eqref{eq:delta-J}\;
        \If{$\Delta J_{i:r\to q}<\Delta J^*$}{
          $\Delta J^*\leftarrow\Delta J_{i:r\to q}$\;
          $q^*\leftarrow q$\;
        }
      }
      \textcolor{gray}{\underline{\textit{Strict descent update:}}}
      \If{$\Delta J^* < 0$}{
        $z_i\leftarrow q^*$\;
        Update $a_r,\mathbf{b}_r,a_{q^*},\mathbf{b}_{q^*}$\;
        For all $i\in G$, decrement $n_r(G)$ and increment $n_{q^*}(G)$\;
      }
    }
  }
  \KwRet{$\mathbf{z}$}\;
\end{algorithm}

This locality comes from two parts of the objective.
After centroids are eliminated, the $k$-means SSE term changes under a single-sample move only through the source and target cluster statistics.
The group penalty term changes only through the constraint groups that contain the moved sample.
Therefore, each candidate reassignment can be evaluated by an exact local delta rather than by recomputing the full objective.
The solver uses two caches to realize this update efficiently: cluster sufficient statistics for the SSE term and group occupancy counts for the ML and CL penalties.
The cache affects only the cost of evaluating a move; accepted moves still strictly decrease the same objective in Eq.~\eqref{eq:objective}.

We first make the $k$-means part local.
For fixed $\mathbf{z}$, the optimal centroid of each cluster is its sample mean.
We maintain two sufficient statistics per cluster,
\begin{equation}
a_\ell = |\{i:z_i=\ell\}|, \qquad
\mathbf{b}_\ell = \sum_{i:z_i=\ell}\mathbf{x}_i,\qquad \ell\in[k].
\label{eq:stats}
\end{equation}
where $a_\ell$ is the cluster size and $\mathbf{b}_\ell$ is its vector sum.
With these cached statistics, the centroid-free SSE in Eq.~\eqref{eq:sse} can be evaluated as $S(\mathbf{z})=C-\sum_{\ell=1}^k\|\mathbf{b}_\ell\|_2^2/a_\ell$.
Evaluating a candidate move therefore requires only the two affected clusters and does not require recomputing all centroids or the full SSE.

For a nontrivial move of sample $i$ from its current cluster $r$ to a target cluster $q\in[k]\setminus\{r\}$, only the source and target statistics change: $a_r'\!=\!a_r-1$ and $\mathbf{b}_r'\!=\!\mathbf{b}_r-\mathbf{x}_i$ for the source cluster, whereas $a_q'\!=\!a_q+1$ and $\mathbf{b}_q'\!=\!\mathbf{b}_q+\mathbf{x}_i$ for the target cluster; all other cluster statistics remain unchanged.
The choice $q=r$ preserves the current assignment and therefore has a zero objective delta, even though the cached statistics may already reflect earlier coordinate updates. 
We include this choice implicitly as the zero-delta baseline and apply the move formulas only to $q\ne r$.
Candidate moves that would empty the source cluster are skipped, so all moves remain in the legal nonempty-cluster assignment space.
Using this cached form of Eq.~\eqref{eq:sse}, the resulting  SSE delta is
\begin{align}
\Delta S_{i:\,r\to q}
\;=\;
\frac{\|\mathbf{b}_r\|_2^2}{a_r} + \frac{\|\mathbf{b}_q\|_2^2}{a_q}
-\frac{\|\mathbf{b}_r - \mathbf{x}_i\|_2^2}{a_r - 1}
-\frac{\|\mathbf{b}_q + \mathbf{x}_i\|_2^2}{a_q + 1}.
\label{eq:move-sse}
\end{align}

We next make the constraint part local.
For each constraint group $G$, we maintain the occupancy vector $(n_1(G),\ldots,n_k(G))$, where $n_\ell(G)=|\{i\in G:z_i=\ell\}|$ records how many members of $G$ currently fall in cluster $\ell$.
A candidate move $i:r\to q$ affects only groups $i\in G$, and for each such group it changes only two entries, namely $n_r(G)\to n_r(G)-1$ and $n_q(G)\to n_q(G)+1$.
Thus the penalty change can be computed from cached counts rather than by rescanning all group members or expanding a group into pairwise constraints.
The penalty deltas follow from these changed entries:
For an ML group, the move changes the penalty by
\[
  \Delta\phi_{\mathrm{ML}}(G;\,i,r\!\to\! q) = \frac{\bigl[\max_\ell n_\ell(G)\bigr]_{\mathrm{old}} - \bigl[\max_\ell n_\ell(G)\bigr]_{\mathrm{new}}}{m-1},
\]
where ``old'' and ``new'' refer to before and after the hypothetical move.

For a CL group, moving $i$ from $r$ to $q$ removes $n_r(G)-1$ collisions in cluster $r$ and adds $n_q(G)$ collisions in cluster $q$, giving
\begin{equation}
\Delta\phi_{\mathrm{CL}}(G;\,i,r\!\to\! q)
= \frac{n_q(G) - (n_r(G)-1)}{\binom{m}{2}}.
\label{eq:cl-delta}
\end{equation}
Eq.~\eqref{eq:cl-delta} can be vectorized over all $k$ target clusters simultaneously, evaluating the full move vector for one group in $O(k)$ time without iterating over group members.

The total move cost combines Eq.~\eqref{eq:move-sse} and the penalty deltas:
\begin{equation}
\Delta J_{i:\,r\to q}
= \Delta S_{i:\,r\to q}
+ \lambda_{\mathrm{ML}}
  \sum_{\substack{G\in\mathcal{M}\\i\in G}}
  \Delta\phi_{\mathrm{ML}}
+ \lambda_{\mathrm{CL}}
  \sum_{\substack{G\in\mathcal{C}\\i\in G}}
  \Delta\phi_{\mathrm{CL}}.
\label{eq:delta-J}
\end{equation}
Defining $\Delta J_{i:\,r\to r}=0$, we select $q^*\in\arg\min_{q\in[k]}\Delta J_{i:\,r\to q}$.
In implementation, this comparison is performed by initializing $q^*=r$ with a zero delta and evaluating Eq.~\eqref{eq:delta-J} only for $q\ne r$.
Because $\Delta J$ is exact, every accepted move strictly decreases the same objective in Eq.~\eqref{eq:objective}.
Following standard incremental coordinate updates~\cite{Nie22Coordinate}, once a move $i:r\to q^*$ is accepted, the cached statistics are updated by
\begin{equation}
\begin{aligned}
a_r &\leftarrow a_r-1,
&\mathbf{b}_r &\leftarrow \mathbf{b}_r-\mathbf{x}_i,\\
a_{q^*} &\leftarrow a_{q^*}+1,
&\mathbf{b}_{q^*} &\leftarrow \mathbf{b}_{q^*}+\mathbf{x}_i,\\
n_r(G) &\leftarrow n_r(G)-1,
&n_{q^*}(G) &\leftarrow n_{q^*}(G)+1.
\end{aligned}
\label{eq:cache-update}
\end{equation}
For sample $i$, these updates take $O(d+r_i)$ work, where $r_i$ is the number of groups containing $i$; averaged over samples, this is $O(d+\bar r)$.
Theorem~\ref{thm:convergence} establishes the convergence of the proposed method in finite steps.

\subsection{Complexity Analysis}
\label{sec:method:complexity}

We analyze the optimization cost after LLM constraints have been generated and stored.
Let $n$ be the number of samples, $d$ be the embedding dimension, and $k$ be the number of clusters.
Let $\mathcal{G}=\mathcal{M}\cup\mathcal{C}$ be the set of all constraint groups, $s=\sum_{G\in\mathcal{G}}|G|$ be the total group-membership size, and $\bar{r}=s/n$ be the average number of incident groups per sample.

Initialization consists of cluster sufficient statistics and the group-count cache.
Building the group-count cache costs $\mathcal{O}(s+k|\mathcal{G}|)$.
During one sweep, all SSE move deltas cost $O(nkd)$ because each candidate move only touches the source and target cluster statistics.
The constraint term visits only groups incident to the moved sample; with dense $k$-dimensional group counts, this costs $\mathcal{O}(n\bar{r}k)$ per sweep.
Thus one sweep costs $\mathcal{O}(nk(d+\bar{r}))$, and $T$ sweeps cost $\mathcal{O}(Tnk(d+\bar{r}))$.
Constraint processing contributes only the additive $\mathcal{O}(Tnk\bar r)$ term, while the $\mathcal{O}(Tnkd)$ SSE computation remains dominant under sparse group incidence.
The optimization cost therefore depends on the average number of incident constraint groups per sample, not on all possible text pairs or all implied pairwise relations inside an ML component.

The state-of-the-art constrained-clustering method \PCCC~\cite{baumann2025pccc} handles constrained assignments through mixed-integer linear programming, whose global solve becomes computationally expensive as the problem and constraint sets grow.
\gpk avoids global integer optimization by evaluating reassignments from cached cluster statistics and incident group counts, yielding a time complexity of $\mathcal{O}(Tnk(d+\bar r))$.
This linear dependence on $n$ avoids the combinatorial growth of global constrained solvers and enables efficient constrained clustering.

\section{Convergence and Perturbation Analysis}
\label{sec:theory}

This section analyzes finite convergence of the cached coordinate descent optimizer and derives objective-level perturbation bounds for noisy LLM-generated constraints.

\gpk optimizes the assignment matrix $\bF$, but the analysis uses the equivalent label vector $\bz$ induced by $\bF$.
Let $B=|\mathcal{M}|+|\mathcal{C}|$ be the total number of typed constraint groups, and enumerate them as $\{(G_t,\tau_t)\}_{t=1}^B$, where $\tau_t\in\{\mathrm{ML},\mathrm{CL}\}$ is the type of group $G_t$.
Write $\phi_t(\bz) = \phi_{\tau_t}(G_t;\bz)$ for the penalty of the $t$-th constraint group.
Let $\lambda_t=\lambda_{\mathrm{ML}}$ if $\tau_t=\mathrm{ML}$ and $\lambda_t=\lambda_{\mathrm{CL}}$ if $\tau_t=\mathrm{CL}$.
The aggregate penalty is $\Phi_\lambda(\bz)=\sum_{t=1}^B \lambda_t\phi_t(\bz)$.
For the noise analysis, $\phi_t^*(\bz)$ denotes a bounded oracle-reference penalty on the same index set $[B]$, and $\widetilde{\phi}_t(\bz)$ denotes the LLM-generated, possibly noisy, penalty.
Their aggregate counterparts are $\Phi_\lambda^*(\bz)=\sum_{t=1}^B\lambda_t\phi_t^*(\bz)$ and $\widetilde{\Phi}_\lambda(\bz)=\sum_{t=1}^B\lambda_t\widetilde{\phi}_t(\bz)$.

\subsection{Boundedness and Special Cases}
\label{sec:theory:bounded}

The critical property enabling bounded regularization is that both penalty functions lie in $[0,1]$.
This limits the maximum objective shift caused by an incorrect LLM group, independent of the group size.

\begin{lemma}[Bounded Influence]
\label{lem:bounded}
For any constraint group $G_t$ with $m = |G_t| \ge 2$ and any assignment $\mathbf{z}$,
\begin{equation}
0 \;\le\; \phi_t(\mathbf{z}) \;\le\; 1.
\label{eq:phi-bounded}
\end{equation}
Consequently, the maximum objective shift caused by a single wrong group is
\begin{equation}
\bigl|\lambda_t\bigl[\widetilde{\phi}_t(\mathbf{z})
- \phi_t^*(\mathbf{z})\bigr]\bigr| \;\le\; \lambda_t.
\label{eq:influence-bound}
\end{equation}
\end{lemma}

\begin{proof}
For an ML group, let $n_\ell(G)=|\{i\in G:z_i=\ell\}|$ denote its occupancy count in cluster $\ell\in[k]$.
The majority cluster count satisfies $1\le \max_{\ell\in[k]}n_\ell(G)\le m$, so $0\le m-\max_{\ell\in[k]}n_\ell(G)\le m-1$.
Dividing by $m-1$ gives $\phi_{\mathrm{ML}}(G;\mathbf{z}) \in [0,1]$.

For a CL group, the numerator $\sum_{\ell\in[k]} \binom{n_\ell(G)}{2}$ counts the number of within-cluster pairs among group members.
This count is at least~$0$ (when all members are in distinct clusters) and at most $\binom{m}{2}$ (when all members are in the same cluster).
Dividing by $\binom{m}{2}$ gives $\phi_{\mathrm{CL}}(G;\mathbf{z}) \in [0,1]$.

Since $\phi_t^*$ and $\widetilde{\phi}_t$ are both in $[0,1]$, their difference satisfies $|\widetilde{\phi}_t - \phi_t^*| \le 1$, and multiplying by $\lambda_t$ yields~\eqref{eq:influence-bound}.
\end{proof}

Equation~\eqref{eq:influence-bound} is the formal counterpart of the key design claim: an incorrect LLM group changes the objective by at most its corresponding ML or CL penalty coefficient, not by an amount that grows with the number of implied pairs.
By contrast, a hard ML merge or a pairwise penalty summed over all $\binom{m}{2}$ pairs has no such finite ceiling.

\begin{lemma}[Pairwise Reduction]
\label{lem:pairwise}
When $|G_t| = 2$, the group penalties in Eqs.~\eqref{eq:ml} and~\eqref{eq:cl} reduce to standard pairwise indicators:
\[
  \phi_{\mathrm{ML}}(\{i,j\};\mathbf{z}) = \mathbb{I}[z_i \ne z_j], \qquad \phi_{\mathrm{CL}}(\{i,j\};\mathbf{z}) = \mathbb{I}[z_i = z_j].
\]
\end{lemma}
\begin{proof}
Write $n_\ell=n_\ell(\{i,j\})$ for $\ell\in[k]$.
For an ML pair, if $z_i = z_j$ then $\max_{\ell\in[k]}n_\ell=2$ and the numerator $m-\max_{\ell\in[k]}n_\ell=0$; if $z_i \ne z_j$ then $\max_{\ell\in[k]}n_\ell=1$ and the numerator equals $1=m-1$.
For a CL pair, $\binom{m}{2}=1$ and $\sum_{\ell\in[k]}\binom{n_\ell}{2}=1$ iff $z_i=z_j$, and it is~$0$ otherwise.
\end{proof}

The group formulation therefore generalizes both ML and CL constraints while recovering their classical pairwise penalties as the special case $m=2$.
The reported experiments use group-valued constraints rather than expanded pairwise penalty terms.

\subsection{Finite Convergence of \gpk}
\label{sec:theory:convergence}

We now show that the coordinate-descent optimizer in Algorithm~\ref{alg:gpkmeans} terminates in finite steps.
This relies on the fact that the group-count cache computes the exact change in the reduced objective, not an approximation.
Here $J$ denotes the fixed generated-constraint objective in Eq.~\eqref{eq:objective}.

Let the legal assignment space be $\mathcal{Z} = \{\mathbf{z}\in[k]^n : a_\ell(\mathbf{z})>0,\;\forall \ell\in[k]\}$, where $a_\ell=|\{i:z_i=\ell\}|$ is the size of cluster~$\ell$.
Expanding the centroid-free SSE term $S(\mathbf{z})$ in Eq.~\eqref{eq:objective} gives the equivalent form
\begin{equation}
  \begin{aligned}
  \!\!J(\mathbf{z})
  \!=C \!-\! \sum_{\ell=1}^k \frac{\|\mathbf{b}_\ell\|_2^2}{a_\ell} 
  \!+\! \lambda_{\mathrm{ML}}\!\!\sum_{G\in\mathcal{M}} \phi_{\mathrm{ML}}(G;\mathbf{z}) \!+\!\lambda_{\mathrm{CL}}\!\!\sum_{G\in\mathcal{C}} \phi_{\mathrm{CL}}(G;\mathbf{z}),
  \end{aligned}
  \label{eq:reduced-obj}
  \end{equation}
where $C$ is a constant and $\mathbf{b}_\ell = \sum_{i:z_i=\ell}\mathbf{x}_i$ is the cluster vector sum for $\ell\in[k]$.

Moving sample $i$ from its current cluster $r$ to a candidate cluster $q\in[k]\setminus\{r\}$ changes $J$ by
\begin{equation}
\Delta J_{i:\,r\to q} = \Delta S_{i:\,r\to q} + \Delta\Phi_{i:\,r\to q}.
\label{eq:theory-delta-J}
\end{equation}
The SSE component $\Delta S$ depends only on the two affected cluster statistics $(a_r, \mathbf{b}_r, a_q, \mathbf{b}_q)$ via Eq.~\eqref{eq:move-sse}.
The penalty component $\Delta\Phi$ depends only on the constraint groups that contain~$i$; all other groups are unaffected.
This locality is what makes the group-count cache exact: for each group $G \ni i$, the cached occupancy vector $\bigl(n_\ell(G)\bigr)_{\ell=1}^k$ changes only in positions $r$ and $q$, and the new penalty values can be evaluated directly from the updated counts without rescanning the group.

\begin{theorem}[Finite Convergence]
\label{thm:convergence}
Suppose a move $(i, r \to q)$ is accepted only when $\Delta J_{i:\,r\to q} < 0$.
Then: the sequence $\{J(\mathbf{z}^{(s)})\}$ of objective values after each accepted move is strictly decreasing; the total number of accepted moves is finite; and the algorithm terminates after finitely many sweeps.
\end{theorem}

\begin{proof}
When a move is accepted, the acceptance criterion gives
\[
  J(\mathbf{z}^{(s+1)}) - J(\mathbf{z}^{(s)}) = \Delta J_{i:\,r\to q} < 0,
\]
so the objective decreases strictly.
When no move is accepted in a sweep, the assignment is unchanged and the objective does not change; the algorithm stops.

The set $\mathcal{Z}$ is finite (it has at most $k^n$ elements).
Since $J$ is a deterministic real-valued function on the finite set $\mathcal{Z}$, it takes only finitely many distinct values.
A strictly decreasing sequence of objective values cannot repeat any value.
Therefore, the number of accepted moves is bounded by $|\mathcal{Z}|$, and the algorithm halts.
\end{proof}

\begin{corollary}[Coordinate-Wise Stationarity]
\label{cor:stationary}
For every sample $i$ whose current cluster $r=z_i^{\mathrm{loc}}$ has $a_r>1$, and every target $q\in[k]\setminus\{r\}$, the terminal assignment $\mathbf{z}^{\mathrm{loc}}$ satisfies
\[
  J(\mathbf{z}^{\mathrm{loc}})
  \;\le\; J\bigl(\mathbf{z}^{\mathrm{loc}}_{i\to q}\bigr),
\]
where $\mathbf{z}^{\mathrm{loc}}_{i\to q}$ is the assignment with sample $i$ reassigned to cluster~$q$.
That is, no legal single-sample move that preserves nonempty clusters can decrease $J$.
\end{corollary}

Theorem~\ref{thm:convergence} establishes optimization stability: the soft group penalties do not introduce cycling or an undefined descent direction.
Importantly, it makes no claim about global optimality.
Since the underlying problem is a discrete, non-convex combinatorial program, the terminal point is a local minimum in the coordinate sense but may be far from a global optimum.

\subsection{LLM Constraint Noise Model}
\label{sec:theory:noise}

Having established bounded penalties and finite descent, we now model the error source that is specific to LLM-constrained clustering: the LLM output may disagree with what the true labels would prescribe.

Let $\widetilde{\phi}_t$ denote the penalty function of the $t$-th final group constructed from the generated constraints, and let $\phi_t^*$ denote a fixed oracle-reference penalty on the same group index set.
Both functions take values in $[0,1]$.
We define the group-level error indicator and its probability as
\begin{equation}
R_t=\mathbb{I}[\widetilde{\phi}_t\ne\phi_t^*],
\qquad q_t=\Pr(R_t=1).
\label{eq:group-error}
\end{equation}
An ML or CL group is clean when all of its induced pairwise relations agree with the label-induced oracle constraints; otherwise, it is treated as a noisy group.
Groups present in only one collection can be represented by the zero function in the other collection, yielding a common index set for comparison.
The bounds below require only these bounded group-level error indicators and do not assume independence unless explicitly stated for the high-probability result.

\subsection{Bounded Penalty Perturbation}
\label{sec:theory:perturbation}

We now quantify how noisy constraints perturb the aggregate penalty.
Define the penalty perturbation at assignment $\mathbf{z}$ as
\begin{equation}
E(\mathbf{z})
= \widetilde{\Phi}_\lambda(\mathbf{z}) - \Phi_\lambda^*(\mathbf{z})
= \sum_{t=1}^{B} \lambda_t \bigl[\widetilde{\phi}_t(\mathbf{z}) - \phi_t^*(\mathbf{z})\bigr].
\label{eq:E-def}
\end{equation}
The term $E(\mathbf{z})$ captures how much the LLM-generated penalty deviates from the oracle penalty at a given assignment.
Because the optimization uses $\widetilde{\Phi}_\lambda$ while the oracle reference is $\Phi_\lambda^*$, bounding $\sup_{\mathbf{z}} |E(\mathbf{z})|$ bounds how much the LLM-generated constraints can distort the objective.

\begin{theorem}[Bounded Penalty Perturbation]
\label{thm:perturbation}
(i) Deterministic bound.
For a fixed set of LLM errors, only the noisy constraints can change the oracle penalty.
Each normalized group penalty can differ from its oracle counterpart by at most one.
Thus the worst-case perturbation over all assignments is bounded by the total penalty coefficient of the erroneous constraints:
\begin{equation}
\sup_{\mathbf{z}} |E(\mathbf{z})| \;\le\; \sum_{t=1}^B \lambda_t R_t.
\label{eq:det-bound}
\end{equation}

(ii) Expected bound.
Taking expectation replaces each binary error indicator $R_t$ by its error probability.
If $\mathbb{E}[R_t]=q_t$, the expected perturbation is controlled by the coefficient-weighted sum of group-level error probabilities:
\begin{equation}
\mathbb{E}\!\left[\sup_{\mathbf{z}} |E(\mathbf{z})|\right]
\;\le\; \sum_{t=1}^B \lambda_t q_t.
\label{eq:exp-bound}
\end{equation}

(iii) High-probability bound.
For any $\delta\in(0,1)$, define
\begin{equation}
\varepsilon_{\mathrm{noise},\delta}
= \sum_{t=1}^B \lambda_t q_t
+ \sqrt{\frac{1}{2}\!\left(\sum_{t=1}^B \lambda_t^2\right)\log\frac{1}{\delta}}.
\label{eq:eps-def}
\end{equation}
Under conditional independence of $R_1,\ldots,R_B$ given the data, Hoeffding's inequality implies that, with probability at least $1-\delta$,
\begin{equation}
\sup_{\mathbf{z}} |E(\mathbf{z})|
\le \varepsilon_{\mathrm{noise},\delta}.
\label{eq:eps-bound}
\end{equation}
\end{theorem}

\begin{proof}
(i) Consider a fixed assignment $\mathbf{z}$ and a fixed realization of the error indicators $\{R_t\}$.
For each group index $t$, if $R_t = 0$, then $\widetilde{\phi}_t = \phi_t^*$ and the $t$-th term in~\eqref{eq:E-def} is exactly~$0$.
If $R_t = 1$, Lemma~\ref{lem:bounded} gives $|\lambda_t[\widetilde{\phi}_t(\mathbf{z}) - \phi_t^*(\mathbf{z})]| \le \lambda_t$.
Summing over all $t$ and noting that the upper bound $\sum_t \lambda_t R_t$ does not depend on $\mathbf{z}$ yields
\[
  |E(\mathbf{z})| \;\le\; \sum_{t=1}^B \lambda_t R_t \quad \text{for all } \mathbf{z},
\]
which is equivalent to~\eqref{eq:det-bound}.

(ii) Apply the tower property: $\mathbb{E}[\sup_\mathbf{z} |E(\mathbf{z})|] \le \mathbb{E}[\sum_t \lambda_t R_t] = \sum_t \lambda_t \mathbb{E}[R_t] = \sum_t \lambda_t q_t$.

(iii) By part~(i), $\sup_\mathbf{z}|E(\mathbf{z})| \le \sum_t \lambda_t R_t$.
Each $\lambda_t R_t$ is a bounded random variable in $[0, \lambda_t]$ with mean $\lambda_t q_t$.
Under conditional independence, applying the Hoeffding inequality to $\sum_t \lambda_t(R_t - q_t)$ gives
\[
  \Pr\!\left(\sum_t \lambda_t R_t \ge \sum_t \lambda_t q_t + u\right) \le \exp\!\left(-\frac{2u^2}{\sum_t \lambda_t^2}\right).
\]
Setting the right-hand side equal to $\delta$ and solving for $u$ yields $u = \sqrt{\tfrac{1}{2}(\sum_t \lambda_t^2)\log(1/\delta)}$, from which $\sup_\mathbf{z}|E(\mathbf{z})| \le \varepsilon_{\mathrm{noise},\delta}$ with probability at least $1-\delta$.
\end{proof}

The bound $\varepsilon_{\mathrm{noise},\delta}$ consists of the expected coefficient-weighted error envelope $\sum_t\lambda_tq_t$ and a Hoeffding deviation term controlling fluctuations above this mean.
The latter scales with $\sqrt{\sum_t\lambda_t^2\log(1/\delta)}$, so the bound grows with the penalty magnitudes and the required confidence level.

\subsection{Oracle Objective Regret}
\label{sec:theory:regret}

To distinguish the generated constraints from their oracle reference in the noise analysis, define the oracle objective $J^*(\mathbf{z}) = S(\mathbf{z}) + \Phi_\lambda^*(\mathbf{z})$ and write the generated-constraint objective $J$ as $\widetilde{J}(\mathbf{z}) = S(\mathbf{z}) + \widetilde{\Phi}_\lambda(\mathbf{z})$.
Their relationship is
\begin{equation}
\widetilde{J}(\mathbf{z})
= J^*(\mathbf{z}) + E(\mathbf{z}),
\label{eq:noisy-obj}
\end{equation}
so a bound on $|E|$ translates directly into a bound on how much $\widetilde{J}$ deviates from $J^*$.
Since coordinate descent guarantees only coordinate-wise stationarity, we present the following comparison with respect to an arbitrary candidate set $\mathcal{H}$ and an explicit optimization error term $\eta_{\mathrm{opt}}$.

\begin{theorem}[Oracle Objective Regret]
\label{thm:regret}
Let $\mathbf{z}_{\mathcal{H}}^*$ minimize $J^*$ over a candidate set $\mathcal{H}$.
Suppose the output $\widehat{\mathbf{z}}\in\mathcal{H}$ satisfies $\widetilde{J}(\widehat{\mathbf{z}})\le\min_{\mathbf{z}\in\mathcal{H}}\widetilde{J}(\mathbf{z})+\eta_{\mathrm{opt}}$.
If $|E(\mathbf{z})|\le\varepsilon$ for all $\mathbf{z}\in\mathcal{H}$, then
\begin{equation}
J^*(\widehat{\mathbf{z}})-J^*(\mathbf{z}_{\mathcal{H}}^*)
\le \eta_{\mathrm{opt}}+2\varepsilon.
\label{eq:regret}
\end{equation}
\end{theorem}

\begin{proof}
Using~\eqref{eq:noisy-obj} and $\sup|E|\le\varepsilon$:
\begin{align*}
J^*(\widehat{\mathbf{z}})
&\le \widetilde{J}(\widehat{\mathbf{z}}) + \varepsilon
  &&\text{(from } \widetilde{J} = J^* + E\text{)}\\
&\le \widetilde{J}(\mathbf{z}_\mathcal{H}^*) + \eta_{\mathrm{opt}}
  + \varepsilon
  &&\text{($\eta_{\mathrm{opt}}$-optimality of }\widehat{\mathbf{z}}\text{)}\\
&\le J^*(\mathbf{z}_\mathcal{H}^*) + \eta_{\mathrm{opt}}
  + 2\varepsilon.
  &&\text{(from } \widetilde{J} = J^* + E\text{ again)}
\end{align*}
Subtracting $J^*(\mathbf{z}_\mathcal{H}^*)$ from both sides gives~\eqref{eq:regret}.
\end{proof}

Applying Theorem~\ref{thm:perturbation}(iii), with probability at least $1-\delta$,
\begin{equation}
J^*(\widehat{\mathbf{z}}) - J^*(\mathbf{z}_\mathcal{H}^*)
\;\le\; \eta_{\mathrm{opt}} + 2\varepsilon_{\mathrm{noise},\delta}.
\label{eq:main-bound}
\end{equation}
This decomposes the total oracle regret into two interpretable terms.
The first, $\eta_{\mathrm{opt}}$, is the optimization error: the gap between the coordinate-wise stationary point found by \gpk and the best assignment in $\mathcal{H}$ under the noisy objective.
The second, $2\varepsilon_{\mathrm{noise},\delta}$, captures noise amplification.
Because $\varepsilon_{\mathrm{noise},\delta}$ is monotone and positively homogeneous in the nonnegative penalty coefficients, stronger penalties enlarge the certified objective distortion caused by noisy groups.
This bound isolates the noise cost of increasing the penalties; the corresponding benefit from clean constraints is data-dependent and is examined empirically in Section~\ref{sec:results}.

\noindent\textit{\underline{\textbf{Remarks.}}}
Once the LLM-generated ML and CL constraints are organized into groups, \gpk solves the constraints clustering problem using cached coordinate descent and Eq.~\eqref{eq:delta-J}.
Overall, the analysis establishes finite-step convergence of this optimizer and quantifies how bounded group penalties control the objective distortion caused by erroneous LLM constraints.

\section{Experiments}
\label{sec:experiments}

The experiments are organized around four questions:
\begin{itemize}[leftmargin=*,nosep]
  \item  Does \gpk improve clustering quality over LLM-based clustering methods and constrained-clustering baselines?
  \item Are the generated constraints sufficiently reliable, and does a higher-quality constraint source improve downstream clustering when the optimizer is held fixed?
  \item Does \gpk exhibit the intended behavior under noisy constraints, such that satisfying more generated constraints does not necessarily yield better clustering?
  \item How many LLM calls does the proposed pipeline require, and how much time does it take?
\end{itemize}
We answer these questions by evaluating clustering quality, analyzing constraint reliability and penalty robustness, and measuring LLM-call counts and running time.
Our code is publicly available. \footnote{\url{https://github.com/NieFeiteng/GPKmeans.git}}

\subsection{Experimental Setup}
\label{sec:setup}

\subsubsection{Datasets.}
We evaluate on nine datasets spanning intent discovery, domain classification, emotion recognition, and social-media topic clustering (Table~\ref{tab:datasets}).
The evaluation includes \Banking~\cite{casanueva2020banking77}, \CLINC~\cite{larson2019clinc}, \Massive~\cite{fitzgerald2023massive}, \MTOP~\cite{li2021mtop}, \GoEmotions~\cite{demszky2020goemotions}, and \Reddit clustering data from the MTEB benchmark~\cite{muennighoff2023mteb}.
These datasets cover both fine-grained intent spaces and coarser domain or topic spaces, allowing us to assess local semantic constraints across different clustering granularities.
All methods use the same embedding file for each dataset, produced by INSTRUCTOR-large~\cite{su2023instructor}.

\begin{table}[t]
  \centering
  \caption{An overview of the datasets.}
  \vspace{-1.4em}
  \label{tab:datasets}
  \resizebox{0.90\linewidth}{!}{\centering
\setlength{\tabcolsep}{13pt} 
\begin{tabular}{lrrr}
\toprule
Dataset & Samples & $k$ & Dim. \\
\midrule
\Banking~\cite{casanueva2020banking77} & 3080 & 77 & 768 \\
\CLINCI~\cite{larson2019clinc} & 4500 & 150 & 768 \\
\MassiveI~\cite{fitzgerald2023massive} & 2974 & 59 & 768 \\
\MTOPI~\cite{li2021mtop} & 4386 & 102 & 768 \\
\CLINCD~\cite{larson2019clinc} & 4500 & 10 & 768 \\
\MassiveD~\cite{fitzgerald2023massive} & 2974 & 18 & 768 \\
\MTOPD~\cite{li2021mtop} & 4386 & 11 & 768 \\
\GoEmo~\cite{demszky2020goemotions} & 5940 & 27 & 768 \\
\Reddit~\cite{muennighoff2023mteb} & 3217 & 50 & 768 \\
\bottomrule
\end{tabular}
}
  \vspace{-1.8em}
\end{table}

\begin{table*}[t]
  \centering
  \scriptsize
  \setlength{\tabcolsep}{3pt}
  \caption{ACC results (\%). Dataset columns report mean $\pm$ standard deviation over five runs, and Avg is the unweighted mean across datasets. Bold and underline denote the best and second-best scores; $\dagger$ marks incomplete \COPKMeans runs.}
  \vspace{-1.9em}
  \label{tab:main-acc}
  \resizebox{0.98\textwidth}{!}{\begin{tabular}{lcccccccccc}
\toprule
Method & \Banking & \CLINCI & \MassiveI & \MTOPI & \CLINCD & \MassiveD & \MTOPD & \GoEmo & \Reddit & Avg \\
\midrule
\Lloyd & 63.83$\pm$1.71 & 79.25$\pm$1.46 & 55.82$\pm$2.33 & 31.74$\pm$1.20 & 50.69$\pm$3.76 & 57.35$\pm$2.62 & 83.32$\pm$5.64 & 13.70$\pm$0.93 & 42.01$\pm$1.64 & 53.08 \\
\LSCK & \underline{68.49$\pm$0.88} & 81.99$\pm$0.47 & 58.22$\pm$0.42 & 33.59$\pm$1.70 & \underline{56.59$\pm$4.83} & \underline{64.60$\pm$2.08} & \textbf{92.98$\pm$0.04} & 14.06$\pm$0.62 & 44.06$\pm$2.49 & \underline{57.18} \\
\LSCKHC & 64.54$\pm$1.43 & 77.88$\pm$1.45 & 54.37$\pm$2.29 & 32.53$\pm$0.83 & 52.97$\pm$4.36 & 59.91$\pm$1.78 & 80.47$\pm$5.30 & 14.14$\pm$0.30 & 43.13$\pm$0.50 & 53.33 \\
\SummaryCentroids & 65.97$\pm$0.48 & 81.27$\pm$0.25 & \textbf{60.13$\pm$1.55} & \textbf{34.88$\pm$1.01} & 55.10$\pm$0.28 & 59.76$\pm$2.08 & 92.70$\pm$0.06 & 14.31$\pm$0.74 & 42.96$\pm$1.08 & 56.34 \\
\PPLL & 65.09$\pm$1.37 & \underline{82.85$\pm$0.74} & 58.79$\pm$0.40 & 34.28$\pm$0.07 & 55.72$\pm$2.01 & 61.41$\pm$3.98 & 92.68$\pm$0.16 & \underline{14.82$\pm$0.19} & \underline{45.65$\pm$0.83} & 56.81 \\
\PCCC & 63.78$\pm$1.20 & 75.74$\pm$1.42 & 54.80$\pm$3.45 & 30.76$\pm$1.23 & 53.52$\pm$1.85 & 55.65$\pm$4.72 & 87.56$\pm$8.66 & 14.40$\pm$0.05 & 41.11$\pm$2.04 & 53.03 \\
\COPKMeans & 56.26$\pm$1.32 & 69.75$\pm$1.31$^{\dagger}$ & 52.88$\pm$0.02$^{\dagger}$ & \underline{34.55$\pm$1.11} & 53.63$\pm$4.60 & 59.26$\pm$0.74 & 82.95$\pm$3.56 & 14.33$\pm$0.44 & N/A$^{\dagger}$ & 52.95$^{\dagger}$ \\
\ours & \textbf{69.33$\pm$0.89} & \textbf{84.34$\pm$0.64} & \underline{60.12$\pm$0.78} & 34.54$\pm$0.66 & \textbf{59.59$\pm$1.71} & \textbf{64.82$\pm$0.28} & \underline{92.84$\pm$0.00} & \textbf{15.51$\pm$0.10} & \textbf{47.62$\pm$0.36} & \textbf{58.75} \\
\bottomrule
\end{tabular}
}
  \vspace{-1.7em}
\end{table*}

\begin{table*}[t]
  \centering
  \scriptsize
  \setlength{\tabcolsep}{3pt}
  \caption{NMI results (\%) under the same protocol as Table~\ref{tab:main-acc}.}
  \vspace{-1.9em}
  \label{tab:main-nmi}
  \resizebox{0.98\textwidth}{!}{\begin{tabular}{lcccccccccc}
\toprule
Method & \Banking & \CLINCI & \MassiveI & \MTOPI & \CLINCD & \MassiveD & \MTOPD & \GoEmo & \Reddit & Avg \\
\midrule
\Lloyd & 82.01$\pm$0.63 & 92.75$\pm$0.22 & 74.36$\pm$0.75 & 67.52$\pm$0.34 & 49.46$\pm$3.01 & 62.59$\pm$1.11 & 83.61$\pm$2.48 & 11.36$\pm$0.76 & 49.81$\pm$0.76 & 63.72 \\
\LSCK & \underline{83.22$\pm$0.45} & 93.14$\pm$0.18 & 74.88$\pm$0.62 & 67.61$\pm$0.54 & 51.67$\pm$2.71 & 64.44$\pm$0.54 & \textbf{88.14$\pm$0.03} & 11.27$\pm$0.35 & 50.57$\pm$1.30 & 64.99 \\
\LSCKHC & 81.82$\pm$0.65 & 92.47$\pm$0.21 & 74.00$\pm$0.34 & 67.64$\pm$0.60 & 49.94$\pm$1.89 & 63.75$\pm$0.17 & 83.32$\pm$1.86 & 10.90$\pm$0.17 & 50.30$\pm$0.35 & 63.79 \\
\SummaryCentroids & 82.65$\pm$0.39 & 93.22$\pm$0.08 & \underline{75.64$\pm$0.12} & 68.13$\pm$0.34 & 51.18$\pm$1.32 & 63.90$\pm$1.32 & 87.89$\pm$0.07 & \textbf{11.71$\pm$0.39} & 49.60$\pm$0.71 & 64.88 \\
\PPLL & 82.14$\pm$0.52 & \underline{93.72$\pm$0.23} & 75.39$\pm$0.40 & \textbf{68.49$\pm$0.33} & \underline{52.45$\pm$1.16} & \underline{65.05$\pm$1.16} & 87.88$\pm$0.24 & 11.38$\pm$0.38 & \underline{51.81$\pm$0.30} & \underline{65.37} \\
\PCCC & 81.49$\pm$0.23 & 90.40$\pm$0.65 & 73.26$\pm$1.81 & 66.80$\pm$0.28 & 51.06$\pm$1.45 & 62.11$\pm$1.48 & 85.16$\pm$4.22 & 11.16$\pm$0.34 & 49.35$\pm$0.97 & 63.42 \\
\COPKMeans & 78.67$\pm$0.50 & 89.94$\pm$0.12$^{\dagger}$ & 72.42$\pm$0.22$^{\dagger}$ & 67.07$\pm$0.15 & 49.66$\pm$1.10 & 63.54$\pm$0.47 & 83.60$\pm$1.79 & 11.21$\pm$0.46 & N/A$^{\dagger}$ & 64.51$^{\dagger}$ \\
\ours & \textbf{83.71$\pm$0.44} & \textbf{94.01$\pm$0.15} & \textbf{75.71$\pm$0.68} & \underline{68.36$\pm$0.39} & \textbf{53.63$\pm$1.38} & \textbf{65.91$\pm$0.27} & \underline{88.12$\pm$0.03} & \underline{11.69$\pm$0.11} & \textbf{52.47$\pm$0.38} & \textbf{65.96} \\
\bottomrule
\end{tabular}
}
  \vspace{-1.7em}
\end{table*}

\begin{table*}[t]
  \centering
  \scriptsize
  \setlength{\tabcolsep}{3pt}
  \caption{ARI results (\%) under the same protocol as Table~\ref{tab:main-acc}.}
  \vspace{-1.8em}
  \label{tab:main-ari}
  \resizebox{0.98\textwidth}{!}{\begin{tabular}{lcccccccccc}
\toprule
Method & \Banking & \CLINCI & \MassiveI & \MTOPI & \CLINCD & \MassiveD & \MTOPD & \GoEmo & \Reddit & Avg \\
\midrule
\Lloyd & 54.50$\pm$1.59 & 75.01$\pm$1.08 & 43.44$\pm$2.52 & 24.25$\pm$1.10 & 33.51$\pm$3.25 & 41.89$\pm$2.24 & 75.69$\pm$5.52 & 3.76$\pm$0.42 & 25.59$\pm$1.24 & 41.96 \\
\LSCK & \underline{58.10$\pm$1.03} & 76.69$\pm$1.04 & 45.38$\pm$0.07 & 26.04$\pm$0.69 & \underline{38.57$\pm$2.94} & \textbf{51.24$\pm$1.60} & \textbf{85.87$\pm$0.18} & 3.77$\pm$0.46 & 27.08$\pm$1.67 & \underline{45.86} \\
\LSCKHC & 54.37$\pm$1.42 & 73.30$\pm$1.00 & 42.48$\pm$1.70 & 24.09$\pm$0.15 & 35.66$\pm$2.94 & 45.46$\pm$1.17 & 74.21$\pm$3.71 & 4.09$\pm$0.17 & 26.48$\pm$0.03 & 42.24 \\
\SummaryCentroids & 56.44$\pm$1.01 & 76.57$\pm$0.39 & \underline{46.76$\pm$1.09} & \textbf{28.90$\pm$2.19} & 36.82$\pm$1.34 & 44.47$\pm$3.11 & 85.37$\pm$0.13 & 4.00$\pm$0.35 & 25.72$\pm$0.48 & 45.01 \\
\PPLL & 54.78$\pm$1.20 & \underline{77.96$\pm$1.28} & 45.41$\pm$0.49 & 26.43$\pm$0.58 & 37.00$\pm$0.56 & 45.29$\pm$4.31 & 85.42$\pm$0.36 & \underline{4.32$\pm$0.50} & \underline{27.08$\pm$0.74} & 44.85 \\
\PCCC & 54.05$\pm$0.54 & 69.76$\pm$1.44 & 42.72$\pm$3.42 & 23.59$\pm$0.14 & 36.06$\pm$0.77 & 39.94$\pm$3.57 & 80.65$\pm$7.59 & 3.99$\pm$0.25 & 24.83$\pm$1.55 & 41.73 \\
\COPKMeans & 47.17$\pm$1.58 & 65.62$\pm$1.01$^{\dagger}$ & 42.67$\pm$0.25$^{\dagger}$ & \underline{27.00$\pm$2.23} & 35.34$\pm$2.02 & 45.63$\pm$0.96 & 76.36$\pm$2.75 & 4.03$\pm$0.55 & N/A$^{\dagger}$ & 42.98$^{\dagger}$ \\
\ours & \textbf{59.11$\pm$0.61} & \textbf{79.32$\pm$0.84} & \textbf{46.85$\pm$1.51} & 26.31$\pm$0.67 & \textbf{39.35$\pm$2.21} & \underline{48.22$\pm$1.12} & \underline{85.61$\pm$0.06} & \textbf{4.39$\pm$0.34} & \textbf{28.36$\pm$0.46} & \textbf{46.39} \\
\bottomrule
\end{tabular}
}
  \vspace{-1.3em}
\end{table*}

\subsubsection{Compared methods.}
We organize the compared methods by the role of the LLM and the source of the constraints.
\Lloyd applies standard Lloyd iterations to the same input embeddings.
\SummaryCentroids~\cite{rodriguez2025kllmmeans} periodically summarizes each current cluster with an LLM and re-embeds the resulting text as a centroid, whereas \PPLL~\cite{zhang2025ppll} uses LLM-generated mini-cluster labels to construct positive pairs, trains an embedder with BYOL, and then performs label-guided $k$-means.
They represent alternative ways of injecting LLM knowledge into clustering.
\LSCK and \LSCKHC~\cite{jia2026optimized} generate set-valued ML and CL constraints with their own LLM pipeline and optimize them using penalty-based local search.
\LSCKHC makes additional LLM calls to divide ML sets into hard and soft constraints: hard sets guide initialization and soft sets are handled through penalties.
These methods therefore compare complete LLM-constraint-generation and clustering pipelines.
We convert the ML and CL constraints generated by \gpk into the corresponding input formats required by \PCCC~\cite{baumann2025pccc} and \COPKMeans~\cite{wagstaff2001constrained}, replacing only the solver.
\PCCC uses an integer-programming-based constrained assignment procedure, whereas \COPKMeans treats the supplied relations as hard feasibility constraints.
\OursLSCKConstraints applies the \gpk optimizer to the constraints generated by \LSCK, thereby isolating the effect of the constraint source while holding the optimizer fixed.
\OracleClean applies the same \gpk optimizer after using true labels to remove generated ML and CL relations that disagree with the label-induced oracle relation.

\subsubsection{Parameter settings.}
Unless otherwise stated, \gpk uses full constraint coverage, group penalties, at most 80 optimization iterations, and penalty coefficients $\lambda_{\mathrm{ML}}=0.01$ and $\lambda_{\mathrm{CL}}=0.1$.
These coefficients are fixed across datasets unless explicitly varied in the sensitivity study.
For sensitivity and controlled studies, we vary only the tested parameter and keep all other settings fixed.
Baseline methods use their corresponding constraints and the same embedding inputs as \gpk. 

\subsubsection{Evaluation metrics.}
We report clustering accuracy (ACC), normalized mutual information (NMI), adjusted Rand index (ARI), constraint violation rates, LLM calls, and runtime.
ACC is computed after Hungarian matching between predicted clusters and true labels.
NMI and ARI evaluate label agreement without requiring a one-to-one cluster assignment.
ACC, NMI, and ARI are reported as percentages.
Constraint violation rates measure how often the final clustering disagrees with the generated ML or CL constraints.
Runtime measures only downstream clustering time; offline LLM generation is reported separately through call counts.

\begin{figure*}[!t]
  \centering
  \includegraphics[width=\textwidth]{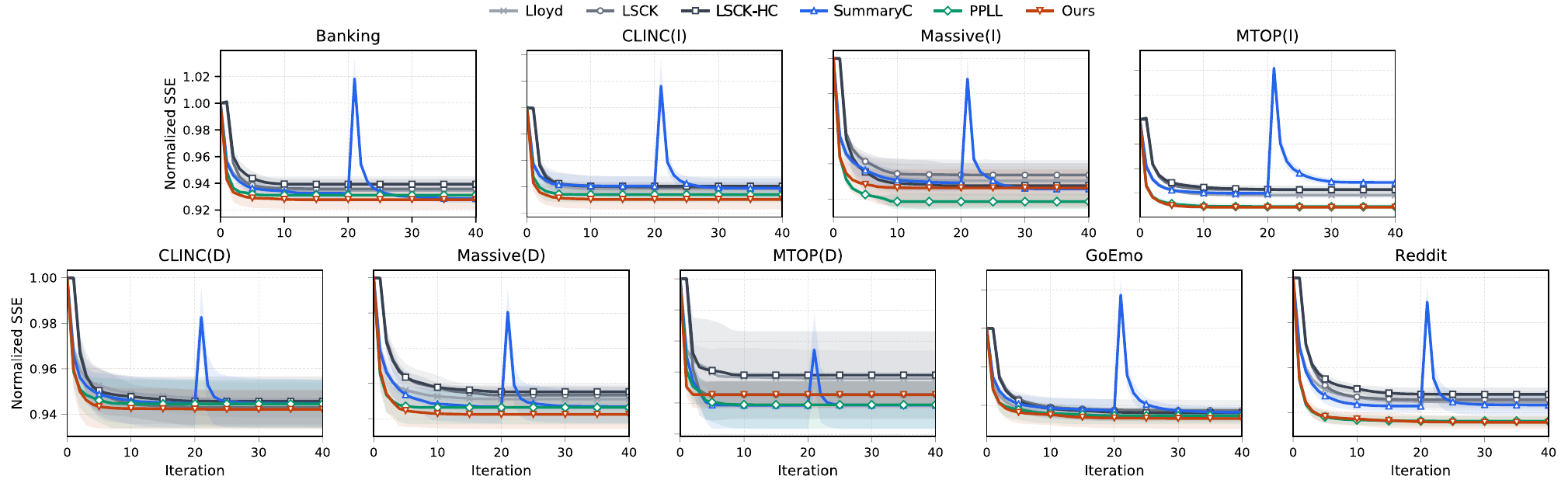}
  \vspace{-2.8em}
  \caption{Normalized SSE over 40 iterations on the nine evaluation datasets.}
  \label{fig:sse}
  \vspace{-1.7em}
\end{figure*}

\subsubsection{Implementation.}
All methods use the same dataset and embedding files. Methods that accept externally supplied constraints use the constraint inputs specified above.
We use five matched random seeds per dataset whenever a method is stochastic and report the mean and standard deviation; $\dagger$ denotes cases with fewer completed runs.
For methods with stored LLM outputs, including \SummaryCentroids and \PPLL, we replay those outputs under the same dataset and embedding protocol.
Runtime-heavy analyses and ablations use matched repeated runs per dataset.
Each solver run has a one-hour timeout.
All LLM constraints are generated offline with \texttt{deepseek-v4-flash} before clustering.
We follow Section~\ref{sec:method:anchor} to generate ML constraints.
For CL constraints, we use the separation sets generated by the procedure of \LSCK~\cite{jia2026optimized}.
The archived ML sources use per-cluster anchor quotas of one, three, or five and evaluate up to 12 candidates per anchor.
This fixed constraint set separates downstream optimizer behavior from the cost and variance of new LLM calls.
All experiments were conducted on a CentOS Linux 7 server (kernel 3.10.0) equipped with two Intel Xeon Platinum 8352V processors and 512~GB of DDR4 ECC memory.

\subsection{Experimental Results}
\label{sec:results}

\begin{table}[t]
  \centering
  \caption{Pair-weighted constraint accuracy (\%).}
  \vspace{-1.4em}
  \label{tab:constraint-quality}
  \resizebox{0.86\linewidth}{!}{\begin{tabular}{lrrrrr}
\toprule
Source & ML Acc. (\%) & CL Acc. (\%) & Overall (\%) & ML Pairs & CL Pairs \\
\midrule
\OracleClean & 100.00 & 100.00 & 100.00 & 8,360 & 34,459 \\
\LSCK & 72.57 & 97.75 & 86.66 & 27,742 & 35,252 \\
\ours & 88.91 & 97.75 & 95.89 & 9,403 & 35,252 \\
\bottomrule
\end{tabular}
}
  \vspace{-1.8em}
\end{table}

\subsubsection{Clustering Quality}
Tables~\ref{tab:main-acc} to~\ref{tab:main-ari} report the main comparison on nine datasets.
The tables report the mean and standard deviation over five runs, except for the incomplete \COPKMeans cases marked by $\dagger$.
Following the final baseline setting, \LSCK and \LSCKHC use penalty coefficient $0.001$.
\SummaryCentroids and \PPLL use stored outputs from their own LLM-based pipelines for the main metric tables; their SSE traces are discussed in Figure~\ref{fig:sse}.
\PCCC and \COPKMeans use the same ML and CL constraints as \gpk, so they test downstream solver behavior under fixed constraints.

Overall, \gpk achieves the best macro average on all three external clustering metrics.
It obtains $58.75\%$ ACC, $65.96\%$ NMI, and $46.39\%$ ARI, improving over the strongest non-\gpk macro average by $1.57\%$, $0.59\%$, and $0.53\%$, respectively.
The gains are moderate rather than uniform across all datasets, as expected, because the embedding space is already strong and the generated constraints are sparse.
The important trend is that the same bounded-penalty formulation improves the macro average across ACC, NMI, and ARI, rather than improving only one matching-sensitive metric.
The \Lloyd row shows that $k$-means is already a strong reference, but it remains below \gpk by $5.67\%$, $2.24\%$, and $4.43\%$ in ACC, NMI, and ARI, respectively.

Compared with \LSCK and \LSCKHC, \gpk gives a more favorable quality-robustness trade-off.
\LSCK remains competitive and obtains the best scores on \MTOPD, showing that hard or high-strength constraints can help when the generated relations align well with the label structure.
However, \gpk is stronger on the macro average and performs especially well on \CLINCI, \CLINCD, \GoEmo, and \Reddit.
This validates treating LLM-generated constraints as semantic preferences rather than strict requirements, since residual noise or redundancy in local constraints can degrade clustering quality when enforced.

The comparison with \SummaryCentroids and \PPLL shows that the advantage of \gpk is not simply due to using an LLM somewhere in the clustering pipeline.
Those methods use LLM-generated summaries or label-oriented guidance, whereas \gpk keeps the LLM role narrow and reusable: it generates local pairwise constraints and leaves the assignment optimization to a regularized $k$-means objective.
The results suggest that such local constraints can be effective when their influence is bounded and optimized jointly with the embedding-based SSE term.

The same-constraint solver baselines provide a more controlled comparison.
Under identical generated ML and CL constraints, \gpk outperforms \PCCC by $5.72\%$, $2.54\%$, and $4.66\%$ in ACC, NMI, and ARI, respectively.
On the eight datasets with reported \COPKMeans results, the matched macro averages favor \gpk by $7.19\%$, $3.13\%$, and $5.67\%$, respectively.
Since the constraint input is fixed in this comparison, the difference is attributable to how the downstream solver uses the constraints.
These results suggest that noisy LLM-generated constraints are better modeled as normalized soft penalties than as hard feasibility constraints.

\begin{figure}[!t]
  \centering
  \resizebox{0.86\linewidth}{!}{
  \includegraphics[width=\linewidth]{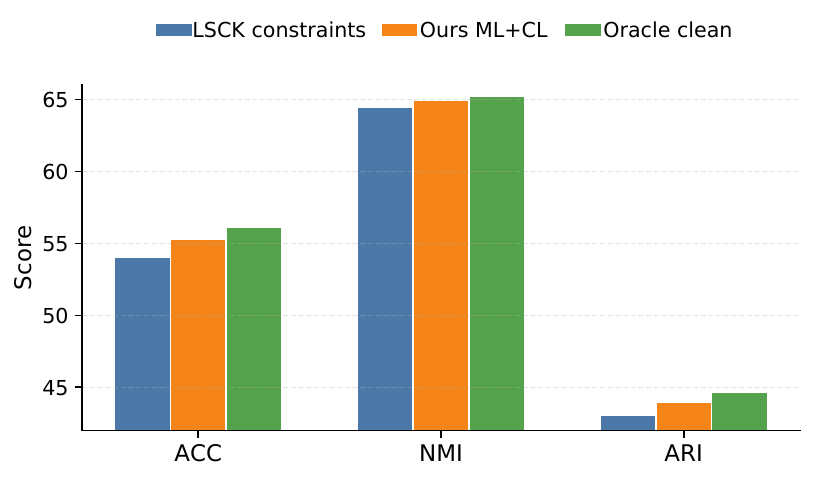}
  }
  \vspace{-1.9em}
  \caption{Effect of constraint source with the \gpk optimizer fixed.}
  \label{fig:constraint-effectiveness}
  \vspace{-2.0em}
\end{figure}

Figure~\ref{fig:sse} reports normalized SSE from the intermediate assignments, with each run's initial value scaled to 1.
For \gpk, \Lloyd, \LSCK, and \LSCKHC, curves show the mean with shaded $\pm 1$ standard deviation across repeated runs.
For \SummaryC and \PPLL, we compute SSE from their stored intermediate assignments.
Because SSE is only one component of the \gpk objective, the curves compare geometric compactness rather than convergence of a shared objective.
\gpk consistently reduces SSE in the first few iterations and usually achieves one of the lowest SSE values among the compared methods.

\subsubsection{Constraint Reliability and Utility}

Table~\ref{tab:constraint-quality} evaluates the generated constraints against true labels over the same nine datasets.
This analysis is important because downstream clustering metrics start with strong embeddings and may show only moderate gains, even when the constraint source is substantially cleaner.
The anchor-based constraint source used by \gpk improves pair-weighted ML accuracy from 72.57\% for \LSCK constraints to 88.91\%.
Because \gpk uses the same set of CL pairs as \LSCK, both methods achieve the same 97.75\% CL accuracy.
The higher-quality ML source raises overall constraint accuracy from 86.66\% to 95.89\%.

This result clarifies the role of anchor-based querying.
It is not a general active-selection method for finding globally optimal query pairs; rather, it provides a sparse, local, and reliable constraint source for the downstream optimizer.
The remaining gap between \gpk and the \OracleClean shows why generated constraints should not be treated as oracle facts: even after improving the ML source, residual invalid constraints remain.

Figure~\ref{fig:constraint-effectiveness} then keeps the optimizer fixed and changes only the constraint input.
The ordering is consistent across ACC, NMI, and ARI: \OracleClean is highest, the deployable \gpk setting is next, and replacing our constraint source with \LSCK constraints is lower.
This comparison should not be read as a comparison against an additional usable baseline; instead, it measures how much of the potential benefit of clean constraints is already recovered by the proposed constraint source.
The small gap between \gpk and \OracleClean suggests that the default constraints capture much of the useful semantic signal.
Further gains may require more targeted boundary-aware selection rather than simply increasing the number of local constraints.

\begin{figure}[t]
  \centering
  \includegraphics[width=\linewidth]{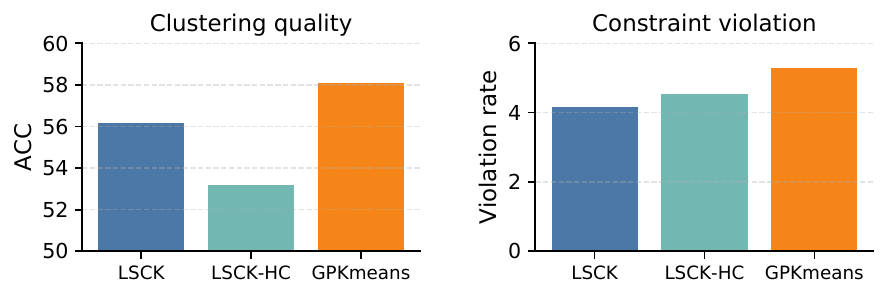}
  \vspace{-2.0em}
  \caption{ACC and violation rates (\%) measured against a common constraint set.}
  \label{fig:soft-penalty}
  \vspace{-1.8em}
\end{figure}

\begin{table}[t]
  \centering
  \caption{Mean scores under increasing oracle constraints.}
\vspace{-1.3em}
  \label{tab:coverage}
  \scriptsize
  \setlength{\tabcolsep}{3pt}
  \resizebox{\linewidth}{!}{\begin{tabular}{lccccccccc}
\toprule
Metric (\%) & 0\% & 10\% & 20\% & 30\% & 40\% & 50\% & 60\% & 70\% & 80\% \\
\midrule
ACC & 62.65 & 67.17 & 70.03 & 73.09 & 76.99 & 79.58 & 87.27 & 94.65 & 100.00 \\
NMI & 84.38 & 86.03 & 86.79 & 88.10 & 89.50 & 90.65 & 94.30 & 97.56 & 100.00 \\
ARI & 45.02 & 51.45 & 54.44 & 59.63 & 64.40 & 68.85 & 80.44 & 91.64 & 100.00 \\
\bottomrule
\end{tabular}
}
  \vspace{-2.5em}
\end{table}
\begin{table*}[t]
  \centering
  \caption{Mean runtime per run (seconds), with a one-hour timeout per run. $\dagger$ marks incomplete \COPKMeans runs.}
  \vspace{-1.4em}
  \label{tab:runtime}
  \resizebox{0.96\textwidth}{!}{\begin{tabular}{lcccccccccc}
\toprule
Method & \Banking & \CLINCI & \MassiveI & \MTOPI & \CLINCD & \MassiveD & \MTOPD & \GoEmo & \Reddit & Avg \\
\midrule
\LSCK & 81.46 & 437.07 & 79.88 & 201.78 & 22.99 & 23.47 & 10.62 & 121.84 & 71.82 & 116.77 \\
\LSCKHC & 45.53 & 250.22 & 66.99 & 284.51 & 15.83 & 24.52 & 9.71 & 152.26 & 74.98 & 102.73 \\
\PCCC & 168.92 & 338.93 & 186.37 & 202.34 & 11.35 & 26.47 & 15.39 & 33.11 & 172.94 & 128.42 \\
\COPKMeans & 501.14 & 1315.16$^{\dagger}$ & 424.64$^{\dagger}$ & 611.69 & 116.70 & 106.51 & 120.72 & 477.75 & N/A$^{\dagger}$ & 459.29$^{\dagger}$ \\
\ours & 2.29 & 6.04 & 3.17 & 3.53 & 1.97 & 2.19 & 1.21 & 7.52 & 2.97 & 3.43 \\
\bottomrule
\end{tabular}
}
  \vspace{-1.5em}
\end{table*}

\begin{figure}[t]
  \centering
  \includegraphics[width=\linewidth]{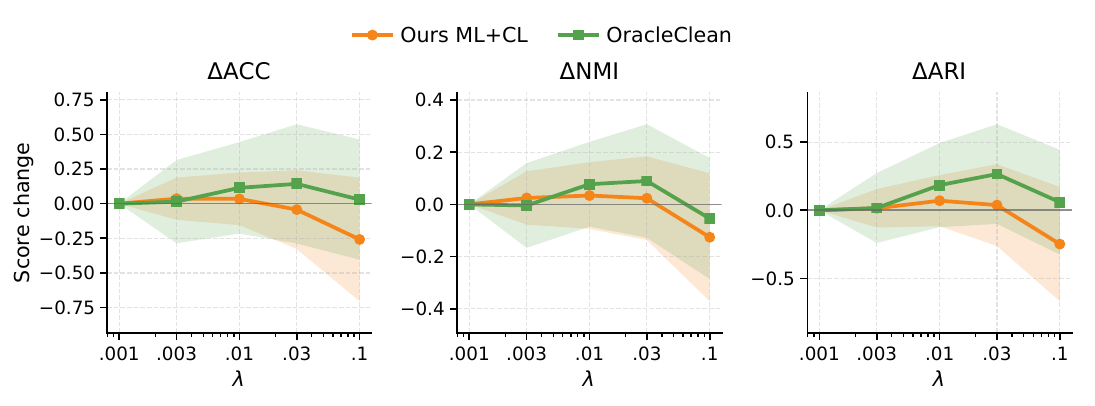}
  \vspace{-2.8em}
  \caption{Sensitivity of \gpk under different values of $\lambda$($\lambda_{\mathrm{ML}}=\lambda_{\mathrm{CL}}=\lambda$).}
  \label{fig:penalty-sensitivity}
  \vspace{-1.8em}
\end{figure}

\subsubsection{Robustness to Noisy Constraints}

Figure~\ref{fig:soft-penalty} examines the relationship between clustering quality and constraint satisfaction by evaluating every method's output against the same generated constraint set.
A lower violation rate is not necessarily better when the constraints are noisy.
\LSCK and \LSCKHC have lower violation rates on this common evaluation set, but their ACC is lower than that of \gpk.
In contrast, \gpk has a higher residual violation rate while achieving the highest ACC in the comparison.
This result is consistent with the proposed formulation: bounded soft penalties can use reliable constraints to correct ambiguous assignments without forcing every noisy relation to be satisfied.

To isolate the effect of constraint coverage, we conduct a controlled study on a balanced \Banking subset.
The subset contains 77 intents with five examples per intent.
We generate the full set of 73,920 oracle ML/CL pairs from gold labels and evaluate nested prefixes of one fixed random ordering from 10\% to 80\%.
The 0\% column in Table~\ref{tab:coverage} reports embedding-only $k$-means++, while the remaining columns report \gpk under increasing oracle-constraint coverage.

Table~\ref{tab:coverage} shows a clear coverage-dependent trend.
Without constraints, $k$-means++ obtains $62.65\%$ ACC, $84.38\%$ NMI, and $45.02\%$ ARI.
Using only 10\% of the oracle constraints improves the three metrics to $67.17\%$, $86.03\%$, and $51.45\%$, respectively.
As coverage increases, the gains become larger: at 60\% coverage, \gpk reaches $87.27\%$ ACC and $80.44\%$ ARI, and at 70\% coverage it reaches $94.65\%$ ACC and $91.64\%$ ARI.
At 80\% coverage, the subset is recovered perfectly.
The observed trend suggests that, as more reliable ML/CL constraints become available, the bounded group-penalty objective is increasingly effective at guiding the clustering structure toward the oracle partition.

Figure~\ref{fig:penalty-sensitivity} reports score changes from $\lambda=0.001$ when both penalty coefficients are set to the common value $\lambda$.
Lines show macro-average changes over nine datasets; the shaded region adds and subtracts the mean within-dataset standard deviation.
The sensitivity study complements the perturbation analysis in Section~\ref{sec:theory:perturbation}: larger penalties increase exposure to erroneous constraints, while their benefit from clean constraints remains data-dependent.
For small and moderate penalty coefficients, the score changes are stable and can improve because reliable constraint relations provide useful semantic guidance.
When the penalty becomes too large, however, clustering quality under the default noisy constraints begins to deteriorate.

The \OracleClean curve is a diagnostic rather than a usable tuning target.
When invalid ML constraints are removed using true labels, the optimizer tolerates stronger penalties better and reaches its best region at a larger coefficient.
The cleaner constraints therefore remain useful over a wider penalty range than the default noisy constraints.
This pattern shows that the useful penalty range depends on constraint reliability, not just on the optimizer.

\subsubsection{LLM Calls and Runtime}

\begin{table}[t]
  \centering
  \caption{Offline LLM calls over nine datasets.}
  \vspace{-1.5em}
  \label{tab:query-efficiency}
  \resizebox{0.86\linewidth}{!}{\begin{tabular}{llr}
\toprule
Method & LLM stage & LLM calls \\
\midrule
\LSCK & ML and CL generation & 5,593 \\
\LSCKHC & ML/CL and confidence estimation & 6,430 \\
\SummaryCentroids & cluster summary updates & 1,512 \\
\PPLL & \PPLL label assignment & 5,400 \\
\ours & Anchor ML + CL generation & 1,435 \\
\bottomrule
\end{tabular}
}
  \vspace{-1.8em}
\end{table}

Table~\ref{tab:query-efficiency} compares the number of LLM calls made by the LLM-based methods over the nine datasets.
We count successful prompt requests reconstructed from cached query logs.
\gpk uses 1,435 calls: 154 calls for anchor-based ML generation and 1,281 calls for CL construction.
In comparison, \LSCK uses 5,593 calls for ML and CL generation, while \LSCKHC adds 837 calls for confidence estimation, for a total of 6,430.
All generated constraints are stored and reused by downstream runs, so this table measures LLM-call overhead rather than solver runtime.

Table~\ref{tab:runtime} reports downstream clustering runtime on the same nine datasets used in Tables~\ref{tab:main-acc}-\ref{tab:main-ari}.
As defined above, this comparison excludes offline prompting cost and focuses on optimization time.
\gpk takes 3.43 seconds per run on average, compared with 116.77 seconds for \LSCK, 102.73 seconds for \LSCKHC, and 128.42 seconds for \PCCC.
Thus, under the offline-constraint protocol, \gpk is about $34.0\times$ faster than \LSCK, $30.0\times$ faster than \LSCKHC, and $37.4\times$ faster than \PCCC.
On the eight datasets with reported \COPKMeans times, \gpk averages 3.49 seconds versus 459.29 seconds for \COPKMeans; several \COPKMeans runs time out or return no feasible assignment.
The speedup comes from local-delta updates.
Each candidate reassignment updates the SSE term using only the source and target cluster statistics, and updates the constraint term using only groups incident to the moved sample.
The group-count cache avoids implied-pair expansion and whole-group rescans, keeping the per-sweep cost close to the sparse form analyzed in Section~\ref{sec:method:complexity}.

\section{Conclusion}
\label{sec:conclusion}

This paper presented \gpk, an efficient Group-Penalty $k$-means method for text clustering with noisy LLM-generated constraints.
\gpk represents these constraints as normalized, bounded group penalties and uses cached coordinate descent to evaluate exact local move deltas without repeatedly computing the full objective or expanding groups into implied pairs.
We established finite-step convergence and derived robustness bounds that characterize the effect of constraint noise under varying ML and CL penalty coefficients.
Across datasets, \gpk achieved relative macro-average improvements of $2.75\%$, $0.90\%$, and $1.16\%$ in ACC, NMI, and ARI over the strongest competing result for each metric, while reducing downstream optimization time by $30\times$ relative to the fastest compared constrained-clustering baseline.

Future work should improve the quality of CL generation and study boundary-aware or utility-guided constraint selection, so that LLM calls focus on constraints that are both reliable and likely to change difficult clustering decisions.
Moreover, we will investigate extending the framework to streaming text clustering, where new samples arrive continuously.

\clearpage
\balance
\bibliographystyle{ACM-Reference-Format}
\bibliography{sample}

\clearpage
\nobalance
\begingroup
\captionsetup[table]{skip=3pt}
\captionsetup[figure]{skip=4pt}
\setlength{\textfloatsep}{8pt plus 1pt minus 1pt}
\setlength{\floatsep}{6pt plus 1pt minus 1pt}
\setlength{\dbltextfloatsep}{8pt plus 1pt minus 1pt}
\setlength{\dblfloatsep}{6pt plus 1pt minus 1pt}

\section*{Technical Report: Adaptive Constraint Acquisition}
\paragraph{Overview.}
$k$-means provides an efficient foundation for text clustering, but embedding proximity can miss task-specific semantic boundaries.
LLM-generated must-link (ML) and cannot-link (CL) constraints provide semantic guidance, but offline acquisition ignores subsequent clustering feedback and noisy judgments can mislead optimization.
We extend \gpk with a feedback loop that couples adaptive constraint acquisition with continued optimization.
Neighborhood disagreement, centroid margins and recent reassignments guide anchor selection from the current partition.
Group--Closure queries within the accumulating anchor set, whereas Group--Fanout extends queries to nonanchor samples.
Accepted relations accumulate into ML connected components and CL clique groups, which enter the objective through normalized, bounded soft penalties.
Cached coordinate descent continues from the current assignment after each update.
Experiments with Gold and real LLM judgments on nine datasets compare both designs under matched total pair-query and sweep budgets.
Updating constraints every two sweeps instead of five improves macro-average ACC by 0.33 and 0.39 percentage points for the LLM Closure and Fanout variants, respectively.
Constraint-quality and graph-redundancy analyses reveal how erroneous and redundant relations limit the gains, motivating more reliable and informative queries.

\section{Extension Design}
\label{tr:sec:design}

\subsection{Purpose and relation to the base method}
This technical report accompanies the preceding \gpk paper.
It supplements the main paper with adaptive constraint acquisition during clustering (Section~\ref{tr:sec:design}) and additional experiments using Gold and real LLM judgments (Section~\ref{tr:sec:evaluation}).

Pairwise supervision~\cite{wagstaff2001constrained} and active constraint acquisition~\cite{basu2004active} motivate our extension.
The base \gpk workflow acquires constraints offline and optimizes a fixed group-penalty objective.
Our extension uses the evolving partition to decide which additional relations to request.
Each acquisition is followed by continued optimization of the current assignment, so information obtained in one interval affects both later assignments and later queries.
The two variants share the same objective and acquisition rules except for the eligible query endpoints.

The extension queries both ML and CL pair relations, rebuilds their groups after each acquisition, and enforces an exact sweep budget.
\gpk uses ML links and separation-set CL construction~\cite{jia2026optimized}, with a different acquisition budget and stopping schedule.
We compare update schedules within the adaptive pipeline at matched total pair-query and sweep budgets.

\subsection{Normalized objective and continued optimization}
Let $\mathbf x_i\in\mathbb R^d$ be the embedding of sample $i$, and let $z_i\in\{1,\ldots,k\}$ be its assignment, with every cluster nonempty.
The centroid $\boldsymbol\mu_c$ is the mean of the embeddings in cluster $c$.
For a group $G$ of size $m\ge2$, let $n_c(G)=|\{i\in G:z_i=c\}|$.
The normalized must-link and cannot-link penalties are
\begin{align}
\phi_{\mathrm{ML}}(G;\mathbf z)&=\frac{m-\max_c n_c(G)}{m-1},\label{tr:eq:ml}\\
\phi_{\mathrm{CL}}(G;\mathbf z)&=\frac{\sum_c\binom{n_c(G)}{2}}{\binom{m}{2}}.\label{tr:eq:cl}
\end{align}
Both terms lie in $[0,1]$ and act as soft preferences; samples remain individually movable.
Writing $S(\mathbf z)=\sum_c\sum_{i:z_i=c}\|\mathbf x_i-\boldsymbol\mu_c\|_2^2$, the objective is
\begin{equation}
\begin{aligned}
J(\mathbf z)&=S(\mathbf z)+\lambda_{\mathrm{ML}}\sum_{G\in\mathcal M}\phi_{\mathrm{ML}}(G;\mathbf z)\\
&\quad+\lambda_{\mathrm{CL}}\sum_{G\in\mathcal C}\phi_{\mathrm{CL}}(G;\mathbf z).
\end{aligned}
\label{tr:eq:objective}
\end{equation}
All runs use $\lambda_{\mathrm{ML}}=0.02$ and $\lambda_{\mathrm{CL}}=0.01$.
The optimization starts from the \gpk $k$-means++ initialization~\cite{arthur2007kmeanspp}, without an additional preliminary run to convergence.
Following incremental coordinate optimization~\cite{Nie22Coordinate}, the solver evaluates exact move costs using cluster statistics and constraint counts, accepts decreasing moves and preserves nonempty clusters.
Samples are visited in the original group-size-based, constrained-first order.

At every acquisition step, the current assignment is retained and the caches are rebuilt for the updated groups.
The solver then executes exactly $\tau$ complete sweeps, where a sweep visits all samples once.
Thus $\tau=2$ means two complete sweeps between acquisitions, rather than two outer acquisition rounds.
Neither labels nor constraints are reset between intervals.
After the last interval, the algorithm returns the current partition directly, without a convergence tail or best-round selection.
The descent guarantee applies within an interval with fixed groups; adding relations and rebuilding groups can change the objective between intervals.

\begin{algorithm}[t]
\caption{Adaptive group-constraint acquisition}
\label{tr:alg:adaptive}
\KwIn{Texts and embeddings, cluster count $k$, initial assignment $\mathbf z$, topology, source, interval $\tau$, and budgets $A$, $S=80$, $B=2560$.}
\KwOut{Final assignment $\mathbf z$.}
$\mathcal A,Q,E_{\mathrm{ML}},E_{\mathrm{CL}}\leftarrow\emptyset$; $c_i\leftarrow0$ for all $i$\;
\For{$s=0,\tau,\ldots,S-\tau$}{
Compute scores from current assignments and movement flags\;
Grow $\mathcal A$ to $\lceil A(s+\tau)/S\rceil$ anchors\;
Select $B\tau/S$ fresh topology-eligible pairs $Q_s$\;
Obtain Gold or LLM judgments; add all $Q_s$ to $Q$\;
Accumulate explicit ML/CL answers; unknown answers add no edge\;
Rebuild ML components and the CL clique partition\;
$\mathbf z_{\mathrm{old}}\leftarrow\mathbf z$\;
Continue normalized \gpk for exactly $\tau$ sweeps\;
$c_i\leftarrow\mathbb I[z_i\ne z_{\mathrm{old},i}]$ for all $i$\;
}
\KwRet{$\mathbf z$}\;
\end{algorithm}

\subsection{Anchor selection from the current partition}
For each point, the selector combines three signals:
\begin{equation}
s_i=r(u_i)+r(-m_i)+0.5c_i.
\label{tr:eq:selection}
\end{equation}
Here $r$ is the empirical rank normalized to $[0,1]$, with average ranks for ties.
The disagreement $u_i$ is the fraction of the point's ten nearest embedding neighbors assigned to other clusters.
Let $a_i$ be its squared distance to its assigned centroid and $b_i$ its squared distance to the nearest competing centroid.
The relative margin is $m_i=(b_i-a_i)/\max\{b_i+a_i,10^{-12}\}$.
The flag $c_i$ indicates a net label change between the start and end of the previous interval and is zero initially.
The score favors neighborhood disagreement, small or negative margins, and recent reassignment.
It is recomputed from the actual current \gpk state before each query batch.

The total anchor budget $A$ is specified as a proportion of the dataset size.
Before an interval beginning after $s$ completed sweeps, the anchor target is $\lceil A(s+\tau)/80\rceil$.
Previously selected anchors are retained.
New anchors are taken from clusters with the fewest existing anchors, prioritizing high-scoring points and using deterministic tie breaking.
The current cluster assignments determine these counts; an anchor can therefore change cluster without leaving the anchor set.
No ground-truth labels are read by the selector.

\subsection{Closure and fanout query topologies}
For Group--Closure, both endpoints of a query must be in the cumulative anchor set.
For Group--Fanout, at least one endpoint must be an anchor, and the other can be any sample.
Both variants visit the ranked anchors in round-robin order and alternate two candidate routes.
The cross-cluster route follows embedding-distance order while requiring a different current cluster.
The within-cluster route prioritizes candidates with a different nearest competing centroid, then higher point scores, shorter distances and sample IDs.
The target split is one half of the batch per route; an exhausted route can be replaced by the other route.
The ten-neighbor set used to compute $u_i$ is a scoring neighborhood, not a hard restriction on the query candidate pool.

Every previously queried unordered pair is excluded, including pairs with uncertain or rejected answers.
The implementation does not assign semantic labels to the entire corpus or derive a dense relation graph from such labels.
Only explicitly queried pairs can create raw ML/CL edges.
Because the source changes the evolving clustering, Gold and LLM runs can select different later pairs and anchors even when they start from the same initial partition.

\subsection{Semantic judgments and cumulative groups}
The Gold oracle returns ML exactly when the selected endpoints have the same benchmark label, and CL otherwise.
The LLM setting makes fresh calls to \texttt{deepseek-v4-flash-0731} at temperature zero, with 32 pairs per main request.
Each pair directly contains both full texts and a local pair ID.
Prompts describe the task's granularity but omit benchmark labels, the official class list and current predicted cluster IDs.
Answers explicitly distinguish ML, CL and uncertainty.
Valid answers are kept; only missing or conflicting format entries receive one repair request, and explicit semantic uncertainty is not repeatedly challenged.
Provider rejections are treated as unknown, and transport failures permit at most three attempts per request.
Raw responses are saved before parsing.
Only identical request payloads within this experiment share cached responses.

All valid returned relations accumulate, including those already satisfied by the current partition.
Unknown answers consume the pair budget but add no constraint.
ML connected components of size at least two become groups under Eq.~\eqref{tr:eq:ml}.
For CL, repeatedly take the lexicographically first residual edge, append the smallest-ID common neighbor until no extension is possible or size $k$ is reached, and remove the resulting clique's edges.
This partitions the observed CL edges into edge-disjoint cliques, including two-point groups.
Cliques can share vertices but never edges, and no additional CL edges are inferred by transitivity.
Each clique uses Eq.~\eqref{tr:eq:cl}.
Rebuilding groups after acquisition can change the normalization of old relations, even though no raw accepted edge is removed.

\section{Extension Experiments}
\label{tr:sec:evaluation}

\subsection{Experimental Setup}
The evaluation includes Banking77~\cite{casanueva2020banking77}, CLINC~\cite{larson2019clinc}, Massive~\cite{fitzgerald2023massive}, MTOP~\cite{li2021mtop}, GoEmotions~\cite{demszky2020goemotions}, and Reddit clustering data from the MTEB benchmark~\cite{muennighoff2023mteb}.
Table~\ref{tr:tab:datasets} summarizes the nine datasets; I/D denote intent/domain clustering.
The main paper uses INSTRUCTOR-large~\cite{su2023instructor} for text representations.
Here, the embeddings are fixed throughout constraint acquisition and optimization.
Matched variants share the same initial partition and use float64 computation.

\begin{table}[!t]
\centering\small
\setlength{\tabcolsep}{3pt}
\caption{Dataset statistics: number of texts $n$, number of classes $k$, and embedding dimension $d$.}
\label{tr:tab:datasets}
\begin{tabular*}{\columnwidth}{@{\extracolsep{\fill}}lrrr@{}}
\toprule
Dataset & $n$ & $k$ & $d$ \\
\midrule
Banking77~\cite{casanueva2020banking77} & 3,080 & 77 & 768 \\
CLINC(I)~\cite{larson2019clinc} & 4,500 & 150 & 768 \\
Massive(I)~\cite{fitzgerald2023massive} & 2,974 & 59 & 768 \\
MTOP(I)~\cite{li2021mtop} & 4,386 & 102 & 768 \\
CLINC(D)~\cite{larson2019clinc} & 4,500 & 10 & 768 \\
Massive(D)~\cite{fitzgerald2023massive} & 2,974 & 18 & 768 \\
MTOP(D)~\cite{li2021mtop} & 4,386 & 11 & 768 \\
GoEmo~\cite{demszky2020goemotions} & 5,940 & 27 & 768 \\
Reddit~\cite{muennighoff2023mteb} & 3,217 & 50 & 768 \\
\bottomrule
\end{tabular*}
\end{table}

The experiment fixes the total query budget at 2,560 distinct pairs and the optimization budget at 80 complete sweeps.
The final anchor budget is approximately 24.35\% of the samples in each dataset.
For $\tau=2$, acquisition occurs before sweeps $0,2,\ldots,78$: 40 batches of 64 pairs.
For $\tau=5$, it occurs before sweeps $0,5,\ldots,75$: 16 batches of 160 pairs.
Sweeps with zero moves still count toward the fixed budget.
Each LLM configuration therefore has 80 nominal main requests before exact-request cache sharing, retries and format repairs.
The evaluation crosses two relation sources (Gold and LLM), two topologies (Closure and Fanout), and two update intervals ($\tau=2$ and $\tau=5$) on each dataset.

The schedules share the same total pair-query and sweep budgets while varying batch size and the timing of anchor release.
We report accepted relations, group counts and model usage alongside clustering quality to characterize the resulting supervision and acquisition cost.
Each configuration is evaluated using the partition saved after its final sweep.
ACC, NMI and ARI are percentages.
Dataset tables report the mean and standard deviation across runs.
Macro results are the unweighted mean of the dataset-level means.

\subsection{Clustering Quality}
We compare the overall clustering quality of the two acquisition topologies and examine how it evolves during optimization.
\begin{table}[!t]
\centering\small
\caption{Macro clustering scores (\%). Closure/Fanout denote the group variants.}
\label{tr:tab:macro}
\begin{tabular*}{\columnwidth}{@{\extracolsep{\fill}}llcrrr@{}}
\toprule
Source & Method & $\tau$ & ACC & NMI & ARI \\
\midrule
Gold & \gpk & -- & 58.54 & 65.92 & 46.18 \\
 & Closure & 5 & 59.38 & 67.06 & 47.69 \\
 & Closure & 2 & 59.58 & 67.26 & 48.07 \\
 & Fanout & 5 & 59.57 & 67.31 & 48.20 \\
 & Fanout & 2 & 59.81 & 67.50 & 48.28 \\
\midrule
LLM & \gpk & -- & 58.75 & 65.96 & 46.39 \\
 & Closure & 5 & 58.54 & 66.11 & 46.38 \\
 & Closure & 2 & 58.86 & 66.30 & 46.80 \\
 & Fanout & 5 & 58.39 & 66.08 & 46.33 \\
 & Fanout & 2 & 58.79 & 66.30 & 46.74 \\
\bottomrule
\end{tabular*}
\end{table}

Table~\ref{tr:tab:macro} shows modest gains from the shorter interval for both sources and both topologies.
Gold ACC increases by 0.203 percentage points for closure and 0.239 for fanout.
With real LLM judgments, the gains are 0.327 and 0.395 points; the corresponding NMI gains are 0.187 and 0.224, and ARI gains are 0.416 and 0.403.
Each ACC comparison improves the dataset mean on six datasets and decreases it on three.
Differences are computed before rounding the displayed scores.

The two LLM variants at $\tau=2$ reach 58.86 and 58.79 ACC, close to the \gpk ACC of 58.75.
The \gpk LLM row uses generated constraints; the Gold row uses oracle answers within the \gpk query scopes and source-specific rules.

\begin{figure*}[!t]
\centering
\includegraphics[width=\textwidth]{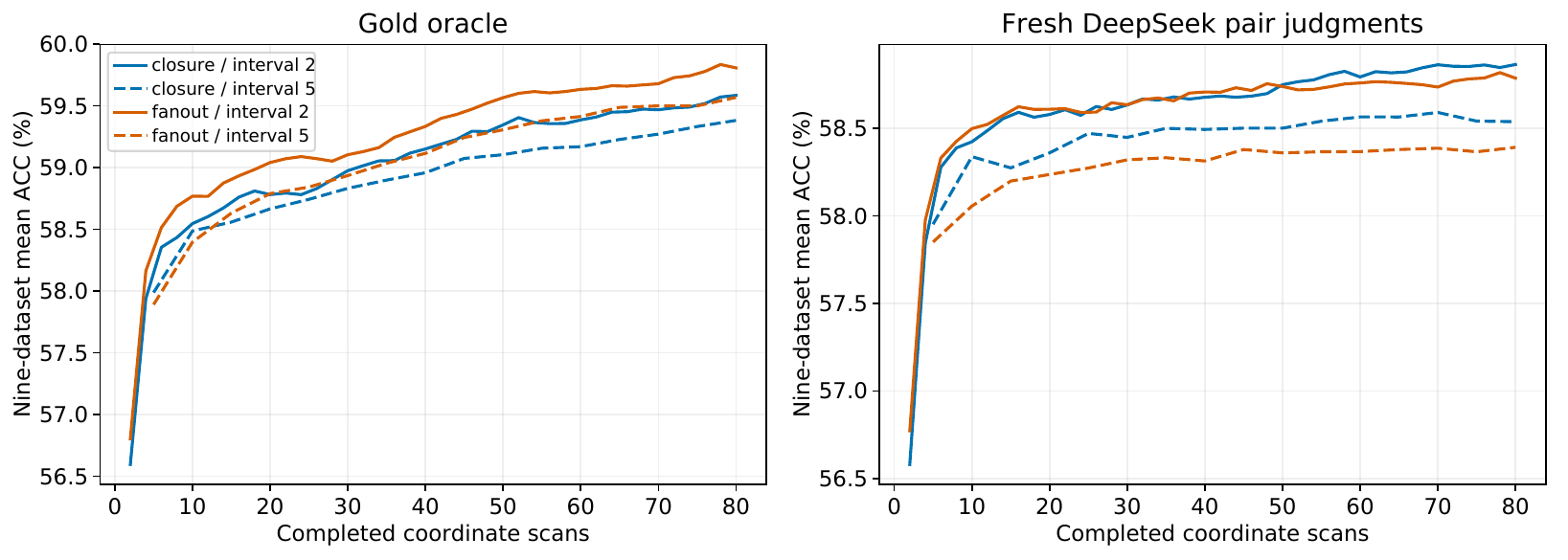}
\caption{Macro ACC over completed sweeps, using saved stage predictions. Lines join saved interval endpoints; all runs end at 80 sweeps and 2,560 queried pairs. Gold and LLM follow their own adaptive trajectories.}
\Description{Two panels plot macro accuracy against completed sweeps for Gold and real LLM judgments. Each panel compares closure and fanout with intervals two and five.}
\label{tr:fig:progress}
\end{figure*}

Figure~\ref{tr:fig:progress} shows rapid early improvement followed by smaller changes as more relations arrive.
Shorter intervals yield more active sweeps, reflecting continued adjustments to newly acquired constraints.
The plotted ACC values are computed from saved predictions for evaluation; query selection uses the current partition and embedding geometry.

\subsection{Effect of Acquisition Frequency}
We examine how shortening the acquisition interval affects clustering quality across datasets.
Tables~\ref{tr:tab:gold-detail} and~\ref{tr:tab:llm-detail} provide all ACC, NMI and ARI means and standard deviations for the adaptive configurations.
In these and the following tables, C/F denote Group--Closure/Group--Fanout and the suffix gives the sweep interval.
The dataset-level results reveal variation in the response to acquisition frequency.

\begin{table}[!t]
\centering\small
\setlength{\tabcolsep}{2.5pt}
\caption{Gold oracle results (\%): mean $\pm$ standard deviation. C/F denote Closure/Fanout; suffixes indicate the sweep interval.}
\label{tr:tab:gold-detail}
\begin{tabular}{@{}llrrrr@{}}
\toprule
Dataset & Metric & C5 & C2 & F5 & F2 \\
\midrule
Banking77 & ACC & $71.13\pm0.35$ & $70.61\pm1.12$ & $71.53\pm0.12$ & $70.88\pm0.31$ \\
 & NMI & $84.91\pm0.33$ & $84.81\pm0.51$ & $85.17\pm0.47$ & $84.96\pm0.75$ \\
 & ARI & $61.24\pm0.24$ & $61.20\pm0.58$ & $62.31\pm0.81$ & $61.56\pm1.13$ \\
\addlinespace[2pt]
CLINC(I) & ACC & $85.30\pm1.54$ & $84.94\pm0.52$ & $85.74\pm0.48$ & $85.77\pm0.64$ \\
 & NMI & $94.96\pm0.32$ & $94.95\pm0.11$ & $95.22\pm0.12$ & $94.87\pm0.08$ \\
 & ARI & $81.50\pm1.77$ & $81.29\pm0.92$ & $82.36\pm0.68$ & $81.82\pm0.53$ \\
\addlinespace[2pt]
Massive(I) & ACC & $60.96\pm1.18$ & $61.37\pm1.46$ & $61.33\pm0.85$ & $61.50\pm2.13$ \\
 & NMI & $77.28\pm0.50$ & $77.32\pm0.86$ & $77.85\pm0.45$ & $77.89\pm0.83$ \\
 & ARI & $47.97\pm1.07$ & $48.44\pm2.11$ & $48.40\pm0.26$ & $48.74\pm1.85$ \\
\addlinespace[2pt]
MTOP(I) & ACC & $34.75\pm0.73$ & $34.97\pm1.05$ & $34.88\pm0.47$ & $34.25\pm0.59$ \\
 & NMI & $68.65\pm0.62$ & $69.04\pm0.44$ & $68.87\pm0.52$ & $68.98\pm0.54$ \\
 & ARI & $26.51\pm1.52$ & $27.12\pm1.43$ & $26.82\pm1.17$ & $26.39\pm1.00$ \\
\addlinespace[2pt]
CLINC(D) & ACC & $59.16\pm2.71$ & $60.33\pm1.49$ & $58.61\pm3.68$ & $60.24\pm1.64$ \\
 & NMI & $54.47\pm1.41$ & $54.95\pm1.35$ & $54.08\pm1.74$ & $55.09\pm1.66$ \\
 & ARI & $39.98\pm2.12$ & $40.42\pm2.05$ & $39.59\pm2.34$ & $40.47\pm2.30$ \\
\addlinespace[2pt]
Massive(D) & ACC & $65.87\pm0.68$ & $65.88\pm0.68$ & $66.11\pm1.10$ & $66.39\pm0.44$ \\
 & NMI & $67.44\pm0.31$ & $67.31\pm0.31$ & $67.87\pm0.50$ & $68.05\pm0.27$ \\
 & ARI & $50.25\pm2.59$ & $50.47\pm1.46$ & $51.52\pm1.78$ & $51.32\pm2.17$ \\
\addlinespace[2pt]
MTOP(D) & ACC & $94.49\pm0.21$ & $95.46\pm0.18$ & $95.01\pm0.10$ & $95.59\pm0.09$ \\
 & NMI & $91.05\pm0.26$ & $92.23\pm0.25$ & $91.82\pm0.15$ & $92.47\pm0.20$ \\
 & ARI & $88.69\pm0.45$ & $90.54\pm0.33$ & $89.64\pm0.22$ & $90.80\pm0.21$ \\
\addlinespace[2pt]
GoEmo & ACC & $14.86\pm0.41$ & $15.07\pm1.04$ & $15.04\pm0.78$ & $15.03\pm0.63$ \\
 & NMI & $11.53\pm0.40$ & $11.64\pm0.24$ & $11.75\pm0.11$ & $11.73\pm0.06$ \\
 & ARI & $4.25\pm0.41$ & $4.22\pm0.55$ & $4.30\pm0.25$ & $4.24\pm0.31$ \\
\addlinespace[2pt]
Reddit & ACC & $47.91\pm0.52$ & $47.63\pm0.36$ & $47.87\pm0.46$ & $48.61\pm0.46$ \\
 & NMI & $53.24\pm0.30$ & $53.11\pm0.19$ & $53.16\pm0.74$ & $53.51\pm0.13$ \\
 & ARI & $28.83\pm0.72$ & $28.97\pm0.72$ & $28.82\pm0.78$ & $29.17\pm0.50$ \\
\bottomrule
\end{tabular}
\end{table}

Gold benefits on CLINC(D), while Banking77 decreases under both topologies.
The variation indicates that acquisition frequency interacts with the embedding geometry, queried relations and task granularity.

\begin{table}[!t]
\centering\small
\setlength{\tabcolsep}{2.5pt}
\caption{Fresh DeepSeek results (\%): mean $\pm$ standard deviation. C/F denote Closure/Fanout; suffixes indicate the sweep interval.}
\label{tr:tab:llm-detail}
\begin{tabular}{@{}llrrrr@{}}
\toprule
Dataset & Metric & C5 & C2 & F5 & F2 \\
\midrule
Banking77 & ACC & $69.71\pm0.56$ & $69.59\pm0.65$ & $70.24\pm0.43$ & $69.70\pm0.33$ \\
 & NMI & $83.76\pm0.40$ & $83.88\pm0.60$ & $84.14\pm0.77$ & $83.73\pm0.76$ \\
 & ARI & $59.37\pm0.61$ & $59.43\pm0.65$ & $60.21\pm0.55$ & $59.43\pm0.90$ \\
\addlinespace[2pt]
CLINC(I) & ACC & $83.84\pm1.19$ & $83.87\pm0.94$ & $84.37\pm0.78$ & $84.55\pm0.71$ \\
 & NMI & $93.96\pm0.22$ & $94.15\pm0.27$ & $94.25\pm0.27$ & $94.25\pm0.19$ \\
 & ARI & $79.16\pm1.45$ & $79.41\pm1.12$ & $79.87\pm1.32$ & $80.15\pm0.85$ \\
\addlinespace[2pt]
Massive(I) & ACC & $60.00\pm0.95$ & $60.32\pm1.44$ & $59.67\pm1.82$ & $60.03\pm1.45$ \\
 & NMI & $76.15\pm0.56$ & $76.38\pm0.48$ & $75.75\pm0.34$ & $76.10\pm0.69$ \\
 & ARI & $46.34\pm1.29$ & $47.06\pm1.66$ & $45.39\pm1.76$ & $46.50\pm1.42$ \\
\addlinespace[2pt]
MTOP(I) & ACC & $34.26\pm0.31$ & $34.69\pm0.71$ & $34.37\pm0.91$ & $33.92\pm0.50$ \\
 & NMI & $68.46\pm0.26$ & $68.45\pm0.30$ & $68.39\pm0.20$ & $68.49\pm0.33$ \\
 & ARI & $26.18\pm1.16$ & $26.67\pm1.24$ & $26.25\pm1.37$ & $26.19\pm0.87$ \\
\addlinespace[2pt]
CLINC(D) & ACC & $58.76\pm2.61$ & $60.71\pm1.75$ & $58.10\pm4.23$ & $59.96\pm1.55$ \\
 & NMI & $53.78\pm1.29$ & $54.38\pm1.38$ & $53.35\pm2.46$ & $54.45\pm1.19$ \\
 & ARI & $39.35\pm2.00$ & $40.09\pm2.31$ & $38.87\pm2.86$ & $40.00\pm1.97$ \\
\addlinespace[2pt]
Massive(D) & ACC & $65.47\pm0.24$ & $65.49\pm0.51$ & $63.99\pm2.71$ & $65.43\pm0.95$ \\
 & NMI & $66.78\pm0.32$ & $66.84\pm0.14$ & $66.41\pm0.92$ & $66.79\pm0.47$ \\
 & ARI & $49.46\pm1.68$ & $49.73\pm1.46$ & $48.46\pm3.50$ & $49.87\pm1.50$ \\
\addlinespace[2pt]
MTOP(D) & ACC & $92.78\pm0.21$ & $93.40\pm0.18$ & $93.14\pm0.46$ & $93.15\pm0.26$ \\
 & NMI & $88.16\pm0.24$ & $88.77\pm0.29$ & $88.65\pm0.78$ & $88.69\pm0.15$ \\
 & ARI & $85.37\pm0.42$ & $86.59\pm0.39$ & $86.00\pm0.87$ & $86.05\pm0.60$ \\
\addlinespace[2pt]
GoEmo & ACC & $14.95\pm0.94$ & $14.81\pm0.43$ & $14.90\pm0.60$ & $14.73\pm0.65$ \\
 & NMI & $11.57\pm0.35$ & $11.44\pm0.29$ & $11.56\pm0.06$ & $11.37\pm0.10$ \\
 & ARI & $4.17\pm0.55$ & $4.07\pm0.22$ & $4.13\pm0.27$ & $4.02\pm0.27$ \\
\addlinespace[2pt]
Reddit & ACC & $47.07\pm0.95$ & $46.91\pm0.05$ & $46.73\pm0.26$ & $47.61\pm0.76$ \\
 & NMI & $52.36\pm0.75$ & $52.38\pm0.20$ & $52.20\pm0.38$ & $52.87\pm0.20$ \\
 & ARI & $28.03\pm0.87$ & $28.13\pm0.60$ & $27.84\pm0.57$ & $28.42\pm0.68$ \\
\bottomrule
\end{tabular}
\end{table}

With LLM judgments, CLINC(D) gains 1.948 ACC points for closure and 1.859 for fanout when moving from five sweeps to two.
Fanout also gains 1.446 points on Massive(D) and 0.881 on Reddit, but loses 0.541 on Banking77 and 0.456 on MTOP(I).
No dataset is excluded from the macro averages.

\subsection{Constraint Reliability and Utility}
We examine relation correctness, graph redundancy and optimization activity to characterize the supervision provided by the acquired constraints.
Table~\ref{tr:tab:counts} distinguishes queried pairs, accepted edges and penalty groups.
Every configuration queries 2,560 pairs, but LLM uncertainty or service rejection can leave some pairs without an edge.
Connected-component ML groups and clique-based CL groups are different objects from the raw queried pairs.
Closure and fanout produce different group counts, which determine the total constraint contribution together with the per-group penalties.

\begin{table}[!t]
\centering\small
\setlength{\tabcolsep}{3pt}
\caption{Final mean counts. Run identifies source, topology and interval. Unknown pairs add no edge; $|\mathcal M|$ and $|\mathcal C|$ count ML/CL groups.}
\label{tr:tab:counts}
\begin{tabular*}{\columnwidth}{@{\extracolsep{\fill}}lrrrrr@{}}
\toprule
Run & \shortstack{ML\\edges} & \shortstack{CL\\edges} & Unknown & $|\mathcal M|$ & $|\mathcal C|$ \\
\midrule
Gold C5 & 695.7 & 1864.3 & 0.0 & 57.0 & 1346.0 \\
Gold C2 & 648.8 & 1911.2 & 0.0 & 50.6 & 1326.4 \\
Gold F5 & 854.6 & 1705.4 & 0.0 & 99.4 & 1414.9 \\
Gold F2 & 779.3 & 1780.7 & 0.0 & 81.7 & 1445.0 \\
LLM C5 & 946.4 & 1601.4 & 12.2 & 9.8 & 1178.2 \\
LLM C2 & 926.9 & 1631.4 & 1.7 & 8.2 & 1141.8 \\
LLM F5 & 973.4 & 1582.2 & 4.4 & 68.1 & 1281.4 \\
LLM F2 & 910.3 & 1643.1 & 6.6 & 54.6 & 1282.0 \\
\bottomrule
\end{tabular*}
\end{table}

Table~\ref{tr:tab:diagnostics} summarizes three diagnostics in addition to relation correctness.
An active sweep makes at least one reassignment.
The violated fraction measures whether each newly acquired valid relation disagrees with the partition \emph{at acquisition time}.
The ML-cycle fraction counts new ML edges that do not reduce the cumulative number of connected components, including cycles formed within a batch.
Since the ML objective depends on components, adding an edge inside an unchanged component alone leaves its group penalty unchanged.
These post-hoc diagnostics characterize the relations selected during acquisition.

\begin{table}[!t]
\centering\small
\setlength{\tabcolsep}{3pt}
\caption{Mean activity and relation quality. Active counts sweeps with moves out of 80; Final counts moves in the last sweep. Remaining columns are percentages, as defined in the text.}
\label{tr:tab:diagnostics}
\begin{tabular*}{\columnwidth}{@{\extracolsep{\fill}}lrrrrrr@{}}
\toprule
Run & Active & Final & Violated & \shortstack{ML\\cycles} & \shortstack{ML\\correct} & \shortstack{CL\\correct} \\
\midrule
Gold C5 & 58.48 & 1.89 & 50.46 & 38.76 & 100.00 & 100.00 \\
Gold C2 & 78.48 & 8.41 & 49.21 & 38.79 & 100.00 & 100.00 \\
Gold F5 & 58.19 & 0.96 & 56.66 & 22.83 & 100.00 & 100.00 \\
Gold F2 & 78.67 & 6.04 & 55.08 & 22.86 & 100.00 & 100.00 \\
LLM C5 & 59.74 & 0.41 & 58.05 & 44.77 & 46.92 & 85.33 \\
LLM C2 & 78.30 & 5.93 & 56.03 & 46.90 & 42.20 & 86.64 \\
LLM F5 & 59.15 & 0.48 & 66.52 & 19.81 & 57.98 & 82.63 \\
LLM F2 & 79.07 & 6.19 & 64.92 & 20.51 & 53.66 & 84.03 \\
\bottomrule
\end{tabular*}
\end{table}

At $\tau=2$, LLM closure and fanout make changes in 78.30 and 79.07 of the 80 sweeps on average, compared with 59.74 and 59.15 at $\tau=5$.
The last sweep moves 5.93 and 6.19 samples on average, indicating continued optimization activity at the budget endpoint.
For the two-sweep LLM runs, ML correctness is 42.20\% for closure and 53.66\% for fanout; CL correctness is 86.64\% and 84.03\%.
Relation correctness is measured against benchmark labels on the queried pairs.
The corresponding ML-cycle fractions are 46.90\% and 20.51\%.
Fanout adds connectivity more often, but its LLM ACC remains slightly below closure's.
New connections can propagate erroneous ML relations, and component merges change group normalization, highlighting the importance of relation reliability alongside connectivity.

\subsection{Constraint Acquisition Cost}
The experiment recorded 13,768,434 tokens, comprising 12,902,134 prompt tokens and 866,300 output tokens.
Token totals are computed from the usage reported by the provider.
The ledger includes format repair and transport retries, with exact-request cache sharing applied across configurations.
Gold consumes no model calls.
Table~\ref{tr:tab:cost} reports the disjoint dataset totals for the valid experiment.
Interval-specific request sets can overlap through cache sharing, so summing their attributed costs would double-count some requests.

\begin{table}[!t]
\centering\small
\setlength{\tabcolsep}{3pt}
\caption{Recorded token consumption by dataset. Prompt and Output count input and generated tokens, respectively; Total is their sum.}
\label{tr:tab:cost}
\begin{tabular*}{\columnwidth}{@{\extracolsep{\fill}}lrrr@{}}
\toprule
Dataset & Prompt & Output & Total \\
\midrule
Banking77 & 1,690,453 & 97,402 & 1,787,855 \\
CLINC(I) & 1,284,448 & 90,213 & 1,374,661 \\
Massive(I) & 1,276,520 & 101,114 & 1,377,634 \\
MTOP(I) & 1,327,318 & 89,108 & 1,416,426 \\
CLINC(D) & 1,300,012 & 89,657 & 1,389,669 \\
Massive(D) & 1,266,427 & 101,082 & 1,367,509 \\
MTOP(D) & 1,312,098 & 96,024 & 1,408,122 \\
GoEmo & 1,758,624 & 100,700 & 1,859,324 \\
Reddit & 1,686,234 & 101,000 & 1,787,234 \\
\midrule
Total & 12,902,134 & 866,300 & 13,768,434 \\
\bottomrule
\end{tabular*}
\end{table}

The final request format was validated with synthetic pairs before benchmark-score analysis.
An earlier format that referenced texts through a separate indexed table answered only 8 of 32 synthetic relations correctly and was stopped.
Providing both texts directly inside each pair answered 32 of 32 correctly, and the exact final request schema passed a separate 32-pair check.
The reported LLM results use only the successful direct-text format; old LLM answers and clustering states were not reused.
Valid Gold checkpoints were reusable because their selection and optimization did not depend on the model's request format.

For complete cost accounting, the indexed-format diagnostic used 10,429,175 recorded tokens, and the synthetic checks used 5,641 tokens.
Including these diagnostic stages, the recorded consumption was 24,203,250 tokens.

\FloatBarrier

\subsection{Discussion}
The experiments support incorporating intermediate clustering feedback into constraint acquisition.
Under matched pair-query and sweep budgets, the two-sweep schedule improves macro-average ACC, NMI and ARI for both topologies and both relation sources.

Topology affects the structure and quality of the resulting supervision.
Fanout reduces redundant ML connections and achieves the highest Gold scores, while closure yields slightly higher LLM ACC at $\tau=2$.
These observations highlight the value of considering relation correctness together with connectivity and group formation.

Relation noise and late constraint updates remain practical challenges: erroneous ML links can propagate through component merges, and the final updates still trigger reassignments.
Future work will investigate reliability-aware selection and adaptive update intervals to obtain informative relations while allowing sufficient optimization after acquisition.

\FloatBarrier
\clearpage
\endgroup

\end{document}